\documentclass[11pt]{article} 
\usepackage
{
        amssymb,
        amsthm,
        amsmath,
        fullpage,
}

\usepackage{harvard}
\usepackage[all,cmtip]{xy}

\DeclareMathOperator {\rank}  {rank}

\DeclareMathOperator {\dem}  {dem}

\DeclareMathOperator {\sgn}  {sgn}

\DeclareMathOperator* {\minext}{min-ext}
\DeclareMathOperator* {\maxflow}{max-flow}

\newcommand {\brc}   [1] {\left(#1\right)}

\newcommand {\Exp}       {\mathbb{E}}
\newcommand {\Prob}  [1] {\Pr \brc{#1 }}
\newcommand {\Probb} [2] {\Pr_{#1} \brc{#2 }}

\newcommand {\EE}    [2] {\Exp_{#1}\left[#2\right]}

\newcommand {\bbR}    {\mathbb{R}}

\newcommand {\calD}   {{\cal{D}}}

\newcommand {\calS}   {{\cal{S}}}

\newcommand {\calP}   {{\cal{P}}}

\newcommand {\calA}   {{\cal{A}}}
\newcommand {\calB}   {{\cal{B}}}
\newcommand {\calL}   {{\cal{L}}}

\newtheorem{theorem}{Theorem}[section]
\newtheorem{lemma}[theorem]{Lemma}

\newtheorem{corollary}[theorem]{Corollary}
\newtheorem{definition}[theorem]{Definition}

\newtheorem{remark}{Remark}[section]

\newtheorem{conjecture}{Conjecture}
\newtheorem{question}[conjecture]{Question}

\title{Metric Extension Operators,  Vertex Sparsifiers and \\ Lipschitz Extendability}
\author{Konstantin Makarychev\\IBM T.J. Watson Research Center
\and Yury Makarychev\\Toyota  Technological Institute at Chicago}
\date{}
\begin{document}
\maketitle
\begin{abstract}
We study vertex cut and flow sparsifiers that were recently introduced 
by \citeasnoun{Moitra}, and \citeasnoun{LM}.
We improve and generalize their results.
We give a new polynomial-time algorithm for constructing $O(\log k / \log \log k)$ cut and flow sparsifiers, matching the best known existential upper bound on the quality of a sparsifier, and improving the previous algorithmic upper
bound of $O(\log^2 k / \log \log k)$. We show that flow sparsifiers can be obtained from 
linear operators approximating minimum metric extensions. We introduce the notion of
(linear) metric extension operators, prove that they exist, and give an exact polynomial-time algorithm for finding optimal operators.

We then establish a direct connection between flow and cut sparsifiers and
Lipschitz extendability of maps in Banach spaces, a notion studied in functional
analysis since 1930s. Using this connection, we obtain a lower bound of
$\Omega(\sqrt{\log k/\log\log k})$ for flow sparsifiers and a lower bound
of $\Omega(\sqrt{\log k}/\log\log k)$ for cut sparsifiers.
We show that if a certain open question posed by Ball in 1992
has a positive answer, then there exist $\tilde O(\sqrt{\log k})$
cut sparsifiers. On the other hand, any lower bound on cut sparsifiers better
than $\tilde \Omega(\sqrt{\log k})$ would imply a negative answer to this question.
\end{abstract}
\section{Introduction}
In this paper, we study \emph{vertex cut and flow sparsifiers} that were recently introduced 
by \citeasnoun{Moitra}, and \citeasnoun{LM}.
A weighted graph $H=(U,\beta)$ is a $Q$-quality vertex cut sparsifier of a weighted
graph $G=(V,\alpha)$ (here $\alpha_{ij}$ and $\beta_{pq}$ are sets of weights on edges of $G$ and $H$) if $U\subset V$ and the size of every 
cut $(S,U\setminus S)$ in $H$ approximates the size of the minimum cut 
separating sets $S$ and $U\setminus S$ in $G$ within a factor of $Q$. 
\citeasnoun{Moitra} presented several important applications of cut sparsifiers
to the theory of approximation algorithms.
Consider a simple example. Suppose we want to 
find the minimum cut in a graph $G=(V,\alpha)$ that splits
a given subset of vertices (terminals) $U\subset V$ 
into two approximately equal parts. We construct $Q$-quality sparsifier $H=(U,\beta)$ of
$G$, and then find a balanced cut $(S, U\setminus S)$ in $H$ using the algorithm of 
\citeasnoun{ARV}. The desired cut is the minimum cut in $G$ separating sets 
$S$ and $U\setminus S$. The approximation ratio we get is $O(Q\times \sqrt{\log |U|})$:
we lose a factor of $Q$ by using cut sparsifiers, and another factor of $O(\sqrt{\log |U|})$ by using
the approximation algorithm for the balanced cut problem. If we applied the approximation algorithm 
for the balanced, or, perhaps, the sparsest cut problem directly we would lose a factor of 
$O(\sqrt{\log |V|})$. This factor depends on the number of vertices in the graph $G$, which may be
much larger than the number of vertices in the graph $H$.
Note, that we gave the example above just to illustrate the method. 
A detailed overview of applications of cut and flow sparsifiers is presented in the papers
of  \citeasnoun{Moitra} and \citeasnoun{LM}.
However, even this simple example shows that 
we would like 
to construct sparsifiers with $Q$ as small as possible. \citeasnoun{Moitra}
proved that for every graph $G=(V,\alpha)$ and every $k$-vertex subset $U\subset V$, there
exists a $O(\log k/\log \log k)$-quality
sparsifier $H=(U,\beta)$. However, the best known polynomial-time 
algorithm proposed by \citeasnoun{LM} finds only 
$O(\log^2 k/\log \log k)$-quality sparsifiers. In this paper, we close this gap: we give a
polynomial-time algorithm for constructing $O(\log k/\log\log k)$-cut sparsifiers matching the best known existential upper bound. In fact, our algorithm constructs $O(\log k/\log\log k)$-\emph{flow sparsifiers}. This type of sparsifiers was introduced by \citeasnoun{LM}; and it generalizes the notion of cut-sparsifiers. Our bound matches the existential upper bound of \citeasnoun{LM} and improves their algorithmic upper bound of $O(\log^2 k/\log \log k)$. If $G$ is a graph with an excluded minor $K_{r,r}$, 
then our algorithm finds a $O(r^2)$-quality flow sparsifier, again matching the best 
existential upper bound of \citeasnoun{LM} (Their algorithmic upper bound has an additional $\log k$ factor). Similarly, we get $O(\log g)$-quality flow
sparsifiers for genus $g$ graphs\footnote{Independently and concurrently to our work, \citeasnoun{CLLM}, and independently \citeasnoun{EGK+}
obtained results similar to some of our results.}.


In the second part of the paper (Section~\ref{sec:lip-extend}), we establish a direct connection between flow and cut sparsifiers and Lipschitz
extendability of maps in Banach spaces.
Let $Q_k^{cut}$ (respectively, $Q_k^{metric}$) be the minimum over all $Q$ such that there exists a $Q$-quality
\textit{cut} (respectively, \textit{flow}) sparsifier for every graph $G=(V,\alpha)$ and every
subset $U\subset V$ of size $k$.
We show that $Q_k^{cut} = e_k(\ell_1, \ell_1)$ and 
$Q_k^{metric} = e_k(\infty, \ell_\infty\oplus_1 \dots\oplus_1\ell_{\infty})$,
where $e_k(\ell_1, \ell_1)$ and $e_k(\infty, \ell_\infty\oplus_1 \dots\oplus_1\ell_{\infty})$
are the \textit{Lipschitz extendability constants} (see Section~\ref{sec:lip-extend} for the definitions).
That is, there always exist cut and flow sparsifiers of quality $e_k(\ell_1, \ell_1)$ and 
$e_k(\infty, \ell_\infty\oplus_1 \dots\oplus_1\ell_{\infty})$, respectively; and these bounds cannot be improved.
We then prove lower bounds on Lipschitz extendability constants and obtain a lower bound of $\Omega(\sqrt{\log k/\log\log k})$ 
on the quality of flow sparsifiers and a lower bound of $\Omega(\sqrt[4]{\log k/\log\log k})$ on the quality of cut 
sparsifiers (improving upon previously known lower bound of $\Omega(\log\log k)$ and $\Omega(1)$ respectively). 
To this end, we employ the connection between Lipschitz extendability constants
and \textit{relative projection constants} that was discovered by \citeasnoun{JL}. 
Our bound on $e_k(\infty, \ell_\infty\oplus_1 \dots\oplus_1\ell_{\infty})$
immediately follows from the bound of \citeasnoun{Grunbaum} on the projection constant
$\lambda(\ell_1^d,\ell_{\infty})$. To get the bound of $\Omega(\sqrt[4]{\log k/\log\log k})$ on $e_k(\ell_1, \ell_1)$, we prove
a lower bound on the projection constant $\lambda(L, \ell_1)$ for a carefully chosen 
subspace $L$ of $\ell_1$. 
After a preliminary version of our paper appeared as a preprint,
Johnson and Schechtman notified us that a lower bound of $\Omega(\sqrt{\log k}/\log\log k)$ on $e_k(\ell_1,\ell_1)$ follows from their joint work with Figiel~\cite{FJS}.
With their permission, we present the proof of the lower bound in Section~\ref{sec:improved-l1} of the Appendix,
which  gives a lower bound of $\Omega(\sqrt{\log k} / \log\log k)$ on the quality of cut sparsifiers.

In Section \ref{subsec:cond-upper-bound}, we note that we can use 
the connection between vertex sparsifiers and extendability constants
not only to prove lower bounds, but also to get positive results.
We show that surprisingly if a certain open question in functional
analysis posed by~\citeasnoun{Ball} has a positive answer, then there exist
$\tilde O(\sqrt{\log k})$-quality cut sparsifiers.
This is both an indication that 
the current upper bound of $O(\log k/\log\log k)$ might not be optimal and
that improving lower bounds beyond of $\tilde O(\sqrt{\log k})$
will require solving a long standing open problem (negatively).

Finally, in Section~\ref{sec:discussion},
we show that there exist simple ``combinatorial certificates''
that certify that $Q_k^{cut} \geq Q$ and $Q_k^{metric} \geq Q$. 

\pagebreak

\textbf{Overview of the Algorithm.}
The main technical ingredient of our algorithm is a procedure for finding 
linear approximations to metric extensions.
Consider a set of points $X$ and a $k$-point subset $Y\subset X$. Let $\calD_X$ be the cone of all metrics on $X$, and $\calD_Y$ be 
the cone of all metrics on $Y$. For a given set of weights $\alpha_{ij}$ on pairs $(i,j)\in X\times X$, the minimum extension of a metric $d_Y$ from $Y$ to $X$ is a metric $d_X$ on $X$ that coincides with $d_Y$ on $Y$ and minimizes the linear functional
$$\alpha(d_X) \equiv \sum_{i,j\in X} \alpha_{ij} d_X(i,j).$$
We denote the minimum above by $\minext_{Y \to X} (d_Y, \alpha)$. We show that the map between $d_Y$ and its minimum extension, the metric $d_X$, can be well approximated by a linear operator. Namely, for every set of nonnegative weights $\alpha_{ij}$ on pairs $(i,j) \in X\times X$, there exists a linear operator
$\phi: \calD_Y \to \calD_X$ of the form
\begin{equation}\label{eq:z}
\phi (d_Y)(i,j) = \sum_{p,q\in Y} \phi_{ipjq} d_Y(p,q)
\end{equation}
that maps every metric $d_Y$ to an extension of the metric $d_Y$ to the set 
$X$ 
such that
$$\alpha(\phi(d_Y)) \leq  O\left(\frac{\log k}{\log \log k}\right) \; 
\minext_{Y \to X} (d_Y, \alpha).$$
As a corollary, the linear functional $\beta: \calD_X \to \bbR$ defined as 
$\beta (d_Y) = \sum_{i,j\in X} \alpha_{ij} \phi(d_Y) (i,j)$ approximates the minimum 
extension of $d_Y$ up to $O(\log k/\log \log k)$ factor. We then give a polynomial-time algorithm 
for finding $\phi$ and $\beta$. (The algorithm finds the optimal $\phi$.) To see the connection with cut and flow sparsifiers
write the linear operator $\beta (d_Y)$ as $\beta (d_Y) = \sum_{p,q\in Y} \beta_{pq} d_Y(p,q)$, then  
\begin{equation}\label{eq:beta}
\minext_{Y\to X} (d_Y, \alpha) \leq 
\sum_{p,q\in Y} \beta_{pq} d_Y(p,q)
\leq O\left(\frac{\log k}{\log \log k}\right) \; \minext_{Y\to X} (d_Y, \alpha).
\end{equation}
Note that the minimum extension of a cut metric is a cut metric (since the mincut
LP is integral).
Now, if $d_Y$ is a cut metric on $Y$ corresponding to the cut $(S, Y\setminus S)$,
then
$\sum_{p,q\in Y} \beta_{pq} d_Y(p,q)$
is the size of the cut in $Y$ with respect to the weights $\beta_{pq}$; and
$\minext_{Y\to X} (d_Y,\alpha)$ is the size of the minimum cut in $X$ separating 
$S$ and $Y\setminus S$. 
Thus, $(Y,\beta)$ is a $O(\log k/\log \log k)$-quality cut sparsifier for $(X, \alpha)$.

\begin{definition}[Cut sparsifier~\cite{Moitra}]
Let $G=(V,\alpha)$ be a weighted undirected graph with weights $\alpha_{ij}$; and let $U\subset V$
be a subset of vertices. We say that a weighted undirected graph $H=(U,\beta)$ on $U$ is a $Q$-quality cut sparsifier, if for every $S\subset U$, 
the size the cut $(S,U\setminus S)$ in $H$ approximates
the size of the minimum cut separating $S$ and $U\setminus S$ in $G$ 
within a factor of $Q$
i.e.,
$$
\min_{T\subset V: S =  T\cap U} 
\sum_{\substack{i\in T\\ j\in V\setminus T}}\alpha_{ij}
\leq
\sum_{\substack{p \in S\\ q \in U\setminus S}} \beta_{pq}
\leq 
Q \times \min_{T\subset V: S =  T\cap U} 
\sum_{\substack{i\in T\\ j\in V\setminus T}}\alpha_{ij}.
$$
\end{definition}

\section{Preliminaries}
In this section, we remind the reader some basic definitions.

\subsection{Multi-commodity Flows and Flow-Sparsifiers}

\begin{definition}
Let $G=(V,\alpha)$ be a weighted graph with nonnegative capacities $\alpha_{ij}$ between vertices $i,j\in V$,
and let $\{(s_r,t_r,\dem_r)\}$ be a set of flow demands ($s_r,t_r\in V$ are terminals of the graph,
$\dem_r\in \bbR$ are demands between $s_r$ and $t_r$; all demands are nonnegative). We say that a weighted collection 
of paths $\calP$ with nonnegative weights $w_p$ ($p\in \calP$) is a fractional multi-commodity flow concurrently satisfying a $\lambda$ fraction of all demands, if the following two conditions hold.
\begin{itemize}
\item Capacity constraints. For every pair $(i,j)\in V\times V$,
\begin{equation}\label{constr:capacity}
\sum_{p\in\calP: (i,j)\in p} w_p \leq \alpha_{ij}.
\end{equation}

\item Demand constraints. For every demand $(s_r,t_r,\dem_r)$,
\begin{equation}\label{constr:demand}
\sum_{p\in\calP: p \text{ goes from  $s_r$ to $t_r$}} w_p \geq \lambda\; \dem_r.
\end{equation}
\end{itemize}
We denote the maximum fraction of all satisfied demands by $\maxflow (G,\{(s_r,t_r,\dem_r)\})$.
\end{definition}

For a detailed overview of multi-commodity flows, we refer the reader to the book of \citeasnoun{Schrijver}.

\begin{definition}[\citeasnoun{LM}]
Let $G=(V,\alpha)$ be a weighted graph and let $U\subset V$ be a subset of vertices. 
We say that a graph $H=(U,\beta)$ on $U$ is a $Q$-quality flow sparsifier of $G$ if
for every set of demands $\{(s_r,t_r,\dem_r)\}$ between terminals in $U$,
$$\maxflow (G,\{(s_r,t_r,\dem_r)\})
\leq
\maxflow (H,\{(s_r,t_r,\dem_r)\})
\leq
Q\times \maxflow (G,\{(s_r,t_r,\dem_r)\}).$$
\end{definition}

\citeasnoun{LM} showed that every flow sparsifier is a cut sparsifier. 

\begin{theorem}[\citeasnoun{LM}]
If $H=(U,\beta)$ is a $Q$-quality flow sparsifier for $G=(V,\alpha)$, then
$H=(U,\beta)$ is also a $Q$-quality cut sparsifier for $G=(V,\alpha)$.
\end{theorem}

\subsection{Metric Spaces and Metric Extensions}

Recall
that a function $d_X:X\times X \to \bbR$ is a metric if for all $i,j, k\in X$ the following
three conditions hold $d_X(i,j) \geq 0$, $d_X(i,j) = d_X (j,i)$, 
$d_X(i,j) + d_X (j,k) \geq d_X(i,k)$. 
Usually, the definition of metric requires that $d_X(i,j) \neq 0$ for distinct 
$i$ and $j$ but we drop this requirement for convenience (such metrics are often
called semimetrics).
We denote the set of all metrics on a set $X$
by $\calD_X$. Note, that $\calD_X$ is a convex closed cone. Moreover, $\calD_X$ is 
defined by polynomially many (in $|X|$) linear constraints (namely, by the three
inequalities above for all $i,j,k\in X$).

A map $f$ from a metric space $(X,d_X)$ to a metric space $(Z,d_Z)$ is 
$C$-Lipschitz, if $d_Z(f(i), f(j)) \leq C d_X(i,j)$ for all $i,j\in X$. The
Lipschitz norm of a Lipschitz map $f$ equals
$$\|f\|_{Lip} = \sup\left\{\frac{d_Z(f(i), f(j))}{d_X(i,j)}: i,j\in X; d_X(i,j) >0\right\}.$$

\begin{definition}[Metric extension and metric restriction]
Let $X$ be an arbitrary set, $Y\subset X$, and $d_Y$ be a metric on $Y$. We say that $d_X$ is a 
metric extension of $d_Y$ to $X$ if $d_X(p,q) = d_Y(p,q)$ for all $p,q \in Y$. If $d_X$ is an 
extension of  $d_Y$, then $d_Y$ is the restriction of $d_X$ to $Y$. We denote
the restriction of $d_X$ to $Y$ by $d_X|_Y$ (clearly, $d_X|_Y$ is uniquely defined by $d_X$).
\end{definition}

\begin{definition}[Minimum extension]
Let $X$ be an arbitrary set, $Y\subset X$, and $d_Y$ be a metric on $Y$. The minimum (cost)
extension of $d_Y$ to $X$ with respect to a set of nonnegative weights $\alpha_{ij}$ on pairs $(i,j)\in X\times X$ 
is a metric extension $d_X$ of $d_Y$ that minimizes the linear functional $\alpha(d_X)$: 
$$\alpha(d_X) \equiv \sum_{i,j\in X} \alpha_{ij} d_X(i,j).$$
We denote $\alpha(d_X)$ by $\minext_{Y \to X} (d_Y, \alpha)$.
\end{definition}

\begin{lemma}
Let $X$ be an arbitrary set,  $Y\subset X$, and  $\alpha_{ij}$ be a set
of nonnegative weights on pairs
$(i,j)\in X\times X$. Then the function $\minext_{Y \to X} (d_Y, \alpha)$ is a convex function of the first variable. 
\end{lemma}
\begin{proof}
Consider arbitrary metrics $d_Y^*$ and $d_Y^{**}$ in $\calD_Y$. Let $d_X^*$ and $d_X^{**}$ be their minimal
extensions to $X$. For every $\lambda\in[0,1]$, the metric $\lambda d_X^* + (1-\lambda) d_X^{**}$ is an
extension (but not necessarily the minimum extension) of $\lambda d_Y^{*} + (1-\lambda) d_Y^{**}$ to $X$,
\begin{multline*}
\minext_{Y\to X} (\lambda d_Y^* + (1-\lambda) d_Y^{**}, \alpha) \leq 
\sum_{i,j\in X} \alpha_{ij}((\lambda d_X^{*} (i,j) + (1-\lambda) d_X^{**} (i,j))) = \\
\lambda \sum_{i,j\in X} \alpha_{ij} d_X^{*} (i,j) + (1-\lambda) \sum_{i,j\in X} \alpha_{ij} d_X^{**} (i,j)
= \lambda \minext_{Y\to X} ( d_Y^*,\alpha) + (1-\lambda) \minext_{Y\to X}(d_Y^{**},\alpha).$$
\end{multline*}
\end{proof}

Later, we shall need the following theorem of \citeasnoun{FHRT}. 

\begin{theorem}[FHRT 0-extension Theorem]\label{thm:FHRT}
Let $X$ be a set of points, $Y$ be a $k$-point subset of $X$, and
$d_Y\in \calD_Y$ be a metric on $Y$. Then for every set of nonnegative 
weights $\alpha_{ij}$ on $X\times X$, there exists a map (0-extension) $f: X \to Y$ such that 
$f(p) = p$ for every $p\in Y$ and 
$$\sum_{i,j\in X} \alpha_{ij} \cdot d_Y(f(i),f(j)) \leq O(\log k/\log\log k)\times 
\minext_{Y\to X} (d_Y,\alpha).$$
\end{theorem}

The notion of 0-extension was introduced by \citeasnoun{Karzanov}. A 
slightly weaker version of this theorem (with a guarantee of $O(\log k)$) was
proved earlier by~\citeasnoun{CKR}. 

\section{Metric Extension Operators}
In this section, we introduce the definitions of
``metric extension operators'' and 
``metric vertex sparsifiers'' and then establish a connection between them
and flow sparsifiers. Specifically, we show that
each $Q$-quality metric sparsifier is
a $Q$-quality flow sparsifier and vice versa (Lemma~\ref{lem:MetricIsFlow}, Lemma~\ref{lem:FlowIsMetric}).
In the next section, we prove that there exist 
metric extension operators with distortion
$O(\log k /\log\log k)$ and give an algorithm that finds the optimal 
extension operator.

\begin{definition}[Metric extension operator]
Let $X$ be a set of points, and $Y$ be a $k$-point subset of $X$. We say that
a linear operator $\phi: \calD_Y \to \calD_X$ defined as
$$\phi(d_Y)(i,j) = \sum_{p,q\in Y} \phi_{ipjq}d_Y(p,q)$$
is a $Q$-distortion metric extension operator 
with respect to a set of nonnegative weights $\alpha_{ij}$, if 
\begin{itemize}
\item for every metric $d_Y\in \calD_Y$, metric $\phi(d_Y)$ is a metric extension of $d_Y$;
\item for every metric $d_Y\in \calD_Y$,
\begin{align*}
\alpha(\phi(d_Y))\equiv&\sum_{i,j\in X} \alpha_{ij} \phi(d_Y)(i,j) \leq  Q\times
\minext_{Y \to X} (d_Y, \alpha).&&\\
\intertext{\textbf{Remark:} As we show in Lemma~\ref{lem:lowerbound}, a stronger bound always holds:}
\minext_{Y \to X} (d_Y, \alpha)\leq&
\alpha (\phi(d_Y))  \leq  Q\times
\minext_{Y \to X} (d_Y, \alpha).&&
\end{align*}
\item for all $i,j\in X$, and $p,q\in Y$, 
$$\phi_{ipjq}\geq 0.$$
\end{itemize}
\end{definition}

We shall always identify the operator $\phi$ with its matrix $\phi_{ipjq}$.

\begin{definition}[Metric vertex sparsifier]
Let $X$ be a set of points, and $Y$ be a $k$-point subset of $X$. We say that
a linear functional $\beta: \calD_Y \to \bbR$ defined as
$$\beta(d_Y) = \sum_{p,q\in Y} \beta_{pq} d_Y(p,q)$$
is a $Q$-quality metric vertex sparsifier 
with respect to a set of nonnegative weights $\alpha_{ij}$, if 
for every metric $d_Y\in \calD_Y$,
$$
\minext_{Y \to X} (d_Y, \alpha) \leq \beta(d_Y) \leq  Q\times
\minext_{Y \to X} (d_Y, \alpha);$$
and all coefficients $\beta_{pq}$ are nonnegative.
\end{definition}

The definition of the metric vertex sparsifier
is equivalent to the definition of the flow vertex sparsifier.
We prove this fact in Lemma~\ref{lem:MetricIsFlow} and 
Lemma~\ref{lem:FlowIsMetric} using duality. However, we shall use the term 
``metric vertex sparsifier'', because the new definition is more 
convenient for us. Also, the notion of metric sparsifiers
makes sense when we restrict $d_X$ and $d_Y$ to be in special 
families of metrics. For example, $(\ell_1,\ell_1)$ metric sparsifiers
are equivalent to cut sparsifiers.
\begin{remark}
The constraints that all $\phi_{ipjq}$ and $\beta_{pq}$ are 
nonnegative though may seem unnatural are required for  
applications. We note that there exist linear operators 
$\phi:\calD_Y \to \calD_X$ and linear functionals $\beta:\calD_Y\to \bbR$
that satisfy all constraints above except for 
the non-negativity constraints. However, even if we drop 
the non-negativity constraints, then there will always exist an optimal 
metric sparsifier with nonnegative constraints (the
optimal metric sparsifier is not necessarily unique).
Surprisingly, the same is not true for metric extension operators:
if we drop the non-negativity constraints, then, in certain cases, the optimal
metric extension operator will necessarily have some negative coefficients. This 
remark is not essential for the further exposition, and we omit the proof here.
\end{remark}

\begin{lemma}\label{lem:lowerbound}
Let $X$ be a set of points, $Y\subset X$, and $\alpha_{ij}$ be a nonnegative set of weights
on pairs $(i,j)\in X\times X$. Suppose that $\phi: \calD_Y \to \calD_X$ is a $Q$-distortion 
metric extension operator. Then
$$\minext_{Y \to X} (d_Y, \alpha)\leq \alpha (\phi(d_Y)).$$
\end{lemma}
\begin{proof}
The lower bound
$$\minext_{Y \to X} (d_Y, \alpha)\leq \alpha (d_X)$$
holds for every extension $d_X$ (just by the definition of the \emph{minimum} metric extension),
and particularly for $d_X = \phi (d_Y)$.
\end{proof}

We now show that given an extension operator with distortion $Q$, it is easy to obtain 
$Q$-quality metric sparsifier.

\begin{lemma}\label{lem:oper2func}
Let $X$ be a set of points, $Y\subset X$, and $\alpha_{ij}$ be a nonnegative set of weights
on pairs $(i,j)\in X\times X$. Suppose that $\phi: \calD_Y \to \calD_X$ is a $Q$-distortion 
metric extension operator. Then there exists a $Q$-quality metric sparsifier $\beta: \calD_Y\to \bbR$.
Moreover, given the operator $\phi$, the sparsifier $\beta$ can be found in polynomial-time.
\end{lemma}
\begin{remark}
Note, that the converse statement does not hold. There exist sets $X$, $Y\subset X$ and weights $\alpha$ such that 
the distortion of the best metric extension operator is strictly larger than
the quality of the best metric vertex sparsifier. 
\end{remark}
\begin{proof}
Let $\beta(d_Y) = \sum_{i,j\in X} \alpha_{ij} \phi(d_Y)(i,j)$. Then by the definition
of $Q$-distortion extension operator, and by Lemma~\ref{lem:lowerbound},
$$\minext_{Y \to X} (d_Y, \alpha)\leq \beta (d_Y) \equiv \alpha(\phi(d_Y)) \leq Q \times \minext_{Y \to X} (d_Y, \alpha).$$
If $\phi$ is given in the form~(\ref{eq:z}), then
$$\beta_{pq} = \sum_{i,j\in X} \alpha_{ij}\phi_{ipjq}.$$
\end{proof}

We now prove that every $Q$-quality metric sparsifier is 
a $Q$-quality flow sparsifier. We prove that every 
$Q$-quality flow sparsifier is a $Q$-quality metric sparsifier
in the Appendix.

\begin{lemma}\label{lem:MetricIsFlow}
Let $G=(V,\alpha)$ be a weighted graph and let $U\subset V$ be a subset of vertices.
Suppose, that a linear functional $\beta: \calD_U \to \bbR$, defined as
$$\beta (d_U) = \sum_{p,q\in U} \beta_{pq} d_U(p,q)$$
is a $Q$-quality metric sparsifier. Then the graph 
$H=(U,\beta)$ is a $Q$-quality flow sparsifier of $G$.
\end{lemma}
\begin{proof}
Fix a set of demands $\{(s_r,t_r,\dem_r)\}$. We need to show, that 
$$\maxflow (G,\{(s_r,t_r,\dem_r)\}) \leq 
\maxflow (H,\{(s_r,t_r,\dem_r)\}) \leq Q\times \maxflow (G,\{(s_r,t_r,\dem_r)\}).$$

The fraction of concurrently satisfied demands by the maximum multi-commodity flow in $G$ equals the maximum of the following standard linear program (LP) for the problem: the LP has a variable $w_p$ for every path between terminals that equals the weight of the path (or, in other words, the amount of flow routed along 
the path) and a variable $\lambda$ that equals the fraction of satisfied demands. The objective is to maximize $\lambda$.  The constraints are the capacity constraints (\ref{constr:capacity}) and demand constraints (\ref{constr:demand}). The maximum of the LP equals the minimum of the (standard) dual LP 
(in other words, it equals the value of the fractional sparsest cut with non-uniform demands). 


\rule{0pt}{12pt}
\hrule height 0.8pt
\rule{0pt}{1pt}
\hrule height 0.4pt
\rule{0pt}{6pt}

\noindent \textbf{minimize:} 
$$\sum_{i,j\in V} \alpha_{ij} d_V(i,j)$$

\noindent \textbf{subject to:}
\begin{align*}
\sum_{r} d_V(s_r,t_r)\times \dem_r &\geq 1&\\ 
d_V & \in \calD_V & \text{i.e., $d_V$ is a metric on } V
\end{align*}
\rule{0pt}{1pt}
\hrule height 0.4pt
\rule{0pt}{1pt}
\hrule height 0.8pt
\rule{0pt}{12pt}

The variables of the dual LP are $d_V(i,j)$, where $i,j\in V$. Similarly, the 
maximum concurrent flow in $H$ equals the minimum of the following dual LP.


\rule{0pt}{12pt}
\hrule height 0.8pt
\rule{0pt}{1pt}
\hrule height 0.4pt
\rule{0pt}{6pt}

\noindent \textbf{minimize:} 
$$\sum_{p,q \in U} \beta_{pq} d_U(p,q)$$

\noindent \textbf{subject to:}
\begin{align*}
\sum_{r} d_U(s_r,t_r)\times \dem_r &\geq 1\\ 
d_U & \in \calD_U & \text{i.e., $d_U$ is a metric on } U
\end{align*}
\rule{0pt}{1pt}
\hrule height 0.4pt
\rule{0pt}{1pt}
\hrule height 0.8pt
\rule{0pt}{12pt}


Consider the optimal solution $d^*_U$ of the dual LP for $H$. Let $d^*_V$ be the 
minimum extension of $d^*_U$. Since $d^*_V$ is a metric, and  $d^*_V (s_r,t_r)
=d^*_U (s_r,t_r)$ for each $r$, $d^*_V$ is a feasible solution of the 
the dual LP for $G$. By the definition of the metric sparsifier:
$$\beta(d^*_U) \equiv \sum_{p,q \in U} \beta_{pq} d^*_U(p,q)
\geq \minext_{Y\to X} (d^*_U,\alpha) \equiv
\sum_{i,j \in V} \alpha_{ij} d^*_V(i,j).$$
Hence, 
$$\maxflow (H,\{(s_r,t_r,\dem_r)\}) \geq \maxflow (G,\{(s_r,t_r,\dem_r)\}).$$

Now, consider the optimal solution $d^*_V$ of the dual LP for $G$. Let $d^*_U$ be the restriction 
of $d^*_V(p,q)$ to the set $U$. Since $d^*_U$ is a metric, and  $d^*_U (s_r,t_r)
=d^*_V (s_r,t_r)$ for each $r$, $d^*_U$ is a feasible solution of the 
the dual LP for $H$. By the definition of the metric sparsifier
 (keep in mind that $d^*_V$ is an extension of $d^*_U$),
$$\beta(d^*_U) \equiv \sum_{p,q \in U} \beta_{pq} d^*_U(p,q)
\leq Q\times \minext_{Y\to X} (d^*_U,\alpha) \leq
Q\times \sum_{i,j \in V} \alpha_{ij} d^*_V(i,j).$$
Hence, 
$$\maxflow (H,\{(s_r,t_r,\dem_r)\}) \leq Q\times \maxflow (G,\{(s_r,t_r,\dem_r)\}).$$
\end{proof}

We are now ready to state the following result.
\begin{theorem}
There exists a polynomial-time algorithm that given a weighted graph $G=(V, \alpha)$ and a
$k$-vertex subset $U\subset V$, finds a $O(\log k/\log\log k)$-quality flow sparsifier $H=(U,\beta)$.
\end{theorem}
\begin{proof}
Using the algorithm given in Theorem~\ref{thrm:alg}, we find the metric extension operator $\phi:\calD_Y\to \calD_X$ with the smallest possible distortion. We output the coefficients of the linear functional $\beta(d_Y) = \alpha(\phi(d_Y))$ (see Lemma~\ref{lem:oper2func}). Hence, by Theorem~\ref{thrm:exists}, the distortion of $\phi$ is at most $O(\log k/\log\log k)$. By Lemma~\ref{lem:oper2func}, $\beta$ is an 
$O(\log k/\log\log k)$-quality metric sparsifier. Finally, by Lemma~\ref{lem:MetricIsFlow}, 
$\beta$ is a $O(\log k/\log\log k)$-quality flow sparsifier (and, thus, a $O(\log k/\log\log k)$-quality cut sparsifier).
 
\end{proof}

\section {Algorithms}

In this section, we prove our main algorithmic results: Theorem~\ref{thrm:exists} and 
Theorem~\ref{thrm:alg}. Theorem~\ref{thrm:exists} asserts that metric extension
operators with distortion $O(\log k/\log\log k)$ exist. To prove 
Theorem~\ref{thrm:exists}, we borrow some ideas from the paper of \citeasnoun{Moitra}.
Theorem~\ref{thrm:alg} asserts that the optimal metric extension
operator can be found in polynomial-time.

Let $\Phi_{Y\to X}$ be the set of all metric extension operators (with 
arbitrary distortion). That is, $\Phi_{Y\to X}$ is the set of 
linear operators $\phi: \calD_Y \to \calD_X$
with nonnegative coefficients $\phi_{ipjq}$ (see (\ref{eq:z}))
that map every metric $d_Y$ on $\calD_Y$ to an extension of $d_Y$ to $X$. 
We show that $\Phi_{Y\to X}$ is closed and convex, and that there exists 
a separation oracle for the set $\Phi_{Y\to X}$.


\begin{corollary}[Corollary of Lemma~\ref{lem:PhiAB} (see below)]~\label{cor:PhiYX}
\begin{enumerate}
\item The set of linear operators $\Phi_{Y\to X}$ is closed and convex.
\item There exists a polynomial-time separation oracle for $\Phi_{Y\to X}$.
\end{enumerate}
\end{corollary}

\begin{lemma}~\label{lem:PhiAB}
Let $\calA\subset \bbR^{m}$ and $\calB\subset \bbR^{n}$ be two polytopes defined by 
polynomially many linear inequalities (polynomially many in $m$ and $n$).
Let $\Phi_{\calA\to\calB}$ be the set of all linear operators $\phi: \bbR^{m} \to \bbR^{n}$,
defined as
$$\phi(a)_i = \sum_{p} \phi_{ip}a_p,$$
that map the set $\calA$ into a subset of $\calB$.
\begin{enumerate}
\item Then $\Phi_{\calA\to\calB}$ is a closed convex set. 
\item There exists a polynomial-time separation oracle for $\Phi_{\calA \to \calB}$. That is,
there exists a polynomial-time algorithm (not depending on $\calA$, $\calB$ and $\Phi_{\calA \to \calB}$),
that given linear constraints for the sets $\calA$, $\calB$, and the $n\times m$ matrix $\phi^*_{ip}$ of a
linear operator $\phi^*: \bbR^{m} \to \bbR^{n}$
\begin{itemize}
\item accepts the input, if $\phi^*\in \Phi_{\calA \to \calB}$.
\item rejects the input, and returns a separating hyperplane, otherwise; i.e., if $\phi^*\notin \Phi_{\calA \to \calB}$, then the oracle returns a linear constraint $l$ such that 
$l(\phi^*) > 0$, but for every $\phi \in \Phi_{\calA\to\calB}$, $l(\phi) \leq 0$.
\end{itemize}
\end{enumerate}
\end{lemma}
\begin{proof}
If $\phi^*,\phi^{**} \in \Phi_{\calA\to\calB}$ and $\lambda \in [0,1]$, then for every $a\in \calA$,
$\phi^*(a) \in \calB$ and $\phi^{**}(a) \in \calB$. Since $\calB$ is convex, 
$\lambda \phi^*(a) + (1-\lambda) \phi^{**}(a)  \in \calB$. Hence, $(\lambda \phi^*+ (1-\lambda) \phi^{**}) (a) \in \calB$. Thus, $\Phi_{\calA\to\calB}$ is convex. If $\phi^{(k)}$ is a Cauchy sequence in
$\Phi_{\calA\to\calB}$, then there exists a limit $\phi = \lim_{k\to \infty} \phi^{(k)}$ and 
for every $a\in \calA$,  $\phi (a) = \lim_{k\to \infty} \phi^{(k)}(a) \in \calB$ (since $\calB$
is closed). Hence, $\Phi_{\calA\to\calB}$ is closed.

Let $\calL_{\calB}$ be the set of linear constraints defining $\calB$:
$$\calB= \{b\in \bbR^n: l(b) \equiv \sum_{i} l_i b_i + l_0\leq 0 \text{ for all } l\in\calL_{\calB}\}.$$
Our goal is to find ``witnesses'' $a\in \calA$ and $l\in \calL_{\calB}$ such that
$l(\phi^*(a)) > 0$. Note that such $a$ and $l$ exist if and only if $\phi^* \notin \Phi$. 
For each $l\in \calL_{\calB}$, write a linear program. The variables of the program are $a_p$,
where $a\in \bbR^m$.


\rule{0pt}{12pt}
\hrule height 0.8pt
\rule{0pt}{1pt}
\hrule height 0.4pt
\rule{0pt}{6pt}

\noindent \textbf{maximize:} 
$l (\phi(a))$

\noindent \textbf{subject to:} $a \in \calA$

\rule{0pt}{1pt}
\hrule height 0.4pt
\rule{0pt}{1pt}
\hrule height 0.8pt
\rule{0pt}{12pt}

This is a linear program solvable in polynomial-time  
since, first, the objective function is a linear function of $a$ (the objective function
is a composition of a linear functional $l$ and a linear operator $\phi$) and, second, 
the constraint $a \in \calA$
is specified by polynomially many linear inequalities.

Thus, if $\phi^*\notin\Phi$, then the oracle
gets witnesses $a^*\in \calA$ and 
$l^*\in \calL_{\calB}$, such that 
$$l^* (\phi^*(a^*)) \equiv \sum_{i}\sum_{p} l^*_i \phi^*_{ip} a_p + l_0>0.$$
The oracle returns the following (violated) linear constraint 
$$l^* (\phi(a^*)) \equiv \sum_{i}\sum_{p} l^*_i \phi_{ip} a_p +l_0\leq 0.$$ 
\end{proof}

\begin{theorem}\label{thrm:exists}
Let $X$ be a set of points, and $Y$ be a $k$-point subset of $X$. For every set of nonnegative weights $\alpha_{ij}$ on $X\times X$, there exists  a metric extension operator $\phi: \calD_Y \to \calD_X$
with distortion $O(\log k/\log\log k)$.
\end{theorem}

\begin{proof}
Fix a set of weights $\alpha_{ij}$. Let 
$\widetilde{\calD}_Y = \{d_Y\in \calD: \minext_{Y \to X} (d_Y, \alpha) \leq 1\}$.
We shall show that there exists $\phi \in \Phi_{Y\to X}$, such that for every $d_Y \in \widetilde{\calD}_Y$
$$\alpha (\phi(d_Y)) \leq O\left(\frac{\log k}{\log \log k}\right),$$
then by the linearity of $\phi$, for every $d_Y\in \calD_Y$
\begin{equation}\label{eq:lesslog}
\alpha (\phi(d_Y))  \leq O\left(\frac{\log k}{\log\log k}\right) \minext_{Y\to X} (d_Y,\alpha).
\end{equation}

The set $\widetilde{\calD}_Y$ is convex and compact, since the function $\minext_{Y \to X} (d_Y, \alpha)$ is a convex function of the first variable. 
The set $\Phi_{Y\to X}$ is convex and closed. Hence, by the \citeasnoun{Neumann} minimax theorem,
$$
\min_{\phi\in \Phi_{Y\to X}} \max_{d_Y\in \widetilde{\calD}_Y} 
\sum_{i,j\in X} \alpha_{ij}\cdot \phi(d_Y)(i,j) = 
\max_{d_Y\in \widetilde{\calD}_Y} \min_{\phi\in \Phi_{Y\to X}} 
\sum_{i,j\in X} \alpha_{ij}\cdot \phi(d_Y)(i,j).$$ 
We will show that the right hand side is bounded by $O(\log k /\log\log k)$, and therefore
there exists $\phi \in \Phi_{Y\to X}$ satisfying (\ref{eq:lesslog}).
Consider $d^*_Y\in \widetilde{\calD}_Y$ for which the maximum above is attained.
By Theorem~\ref{thm:FHRT} (FHRT 0-extension Theorem), there exists 
a 0-extension $f:X\to Y$ such
that $f(p) = p$ for every $p\in Y$, and 
$$\sum_{i,j\in X} \alpha_{ij}\cdot d^*_Y(f(i), f(j))
 \leq  O\left(\frac{\log k}{\log\log k}\right) \minext_{Y\to X} (d^*_Y,\alpha)
\leq O\left(\frac{\log k}{\log\log k}\right).$$ 
Define $\phi^* (d_Y) (i,j) =  d_Y(f(i), f(j))$. Verify that $\phi^* (d_Y)$
is a metric for every $d_Y \in \calD_Y$:
\begin{itemize}
\item $\phi^* (d_Y) (i,j) = d_Y(f(i), f(j)) \geq 0;$
\item $\phi^* (d_Y) (i,j) + \phi^* (d_Y) (j,k) - \phi^* (d_Y) (i,k) = 
d_Y(f(i), f(j)) + d_Y(f(j), f(k)) - d_Y(f(i), f(k)) \geq 0$.
\end{itemize}
Then, for $p,q \in Y$, $\phi^* (d_Y)(p,q) = d_Y(f(p),f(q)) = d_Y(p,q)$, hence 
$\phi^* (d_Y)$ is an extension of $d_Y$. All coefficients $\phi^*_{ipjq}$ of  
$\phi^*$ (in the matrix representation (\ref{eq:z})) equal 0 or 1.
Thus, $\phi^* \in \Phi_{Y\to X}$. Now,
$$\sum_{i,j\in X} \alpha_{ij}\cdot \phi^*(d^*_Y)(i,j) = 
\sum_{i,j\in X} \alpha_{ij}\cdot d^*_Y(f(i),f(j)) \leq  
O\left(\frac{\log k}{\log\log k}\right).$$
This finishes the the proof, that there exists $\phi \in \Phi_{Y\to X}$ satisfying the upper bound
(\ref{eq:lesslog}).
\end{proof}

\begin{theorem}
Let $X$, $Y$, $k$, and $\alpha$ be as in Theorem~\ref{thrm:exists}. Assume further, that
for the given $\alpha$ and every metric $d_Y\in \calD_Y$, there exists a $0$-extension $f:X\to Y$
such that 
$$\sum_{i,j\in X} \alpha_{ij}\cdot d_Y(f(i), f(j)) \leq Q\times \minext_{Y\to X}(d_Y, \alpha).$$ 
Then there exists a metric extension operator with distortion $Q$. Particularly, if the support of the
weights $\alpha_{ij}$ is a graph with an excluded minor $K_{r,r}$, then $Q=O(r^2)$.
If the graph $G$ has genus $g$, then $Q=O(\log g)$.
\end{theorem}
The proof of this theorem is exactly the same as the proof of Theorem~\ref{thrm:exists}.
For graphs with an excluded minor we use a result of \citeasnoun{CKR} (with 
improvements by~\citeasnoun{FT}). For graphs of genus $g$, we use a result of
\citeasnoun{LS}.

\begin{theorem}\label{thrm:alg}
There exists a polynomial time algorithm that given a set of points $X$, a $k$-point subset $Y\subset X$, and a set of positive weights $\alpha_{ij}$, finds a metric extension operator $\phi:\calD_Y\to\calD_X$ with 
the smallest possible distortion $Q$.
\end{theorem}
\begin{proof}
In the algorithm, we represent the linear operator $\phi$ as a matrix $\phi_{ipjq}$ (see (\ref{eq:z})). To find optimal $\phi$, we write a convex program with variables $Q$ and $\phi_{ipjq}$:

\rule{0pt}{12pt}
\hrule height 0.8pt
\rule{0pt}{1pt}
\hrule height 0.4pt
\rule{0pt}{6pt}

\noindent \textbf{minimize:} $Q$

\noindent \textbf{subject to:}
\begin{align}
\alpha(\phi(d_Y)) &\leq Q\times \minext_{Y\to X}(d_Y,\alpha),& \text{ for all } d_Y \in \calD_Y\label{constr:1}\\ 
\phi &\in \Phi_{Y\to X}&\label{constr:2}
\end{align}
\rule{0pt}{1pt}
\hrule height 0.4pt
\rule{0pt}{1pt}
\hrule height 0.8pt
\rule{0pt}{12pt}

The convex problem exactly captures the definition of the extension operator. Thus the solution 
of the program corresponds to the optimal $Q$-distortion extension operator. However, a priori, it
is not clear if this convex program can be solved in polynomial-time. It has exponentially many linear constraints of type (\ref{constr:1}) and one convex non-linear constraint  (\ref{constr:2}). We already know (see Corollary~\ref{cor:PhiYX}) that there exists a separation oracle for $\phi \in \Phi_{Y\to X}$.
We now give a separation oracle for constraints~(\ref{constr:1}).

\medskip

\noindent \textbf{Separation oracle for~(\ref{constr:1}).} The goal of the oracle is 
given a linear operator $\phi^*: d_Y \mapsto \sum_{p,q}\phi^*_{ipjq} d_Y(p,q)$ and a real number $Q^*$ 
find a metric $d^*_Y \in \calD_Y$, such that the constraint 
\begin{equation}\label{eq:dY}
\alpha(\phi^*(d^*_Y)) \leq Q^*\times \minext_{Y\to X}(d^*_Y,\alpha)
\end{equation}
is violated. We write a linear program on $d_Y$. However, 
instead of looking for a metric $d_Y\in \calD_Y$ such that constraint~(\ref{eq:dY}) is violated,
we shall look for a metric $d_X\in \calD_X$, an arbitrary metric extension of $d_Y$ to $X$, such that
$$\alpha(\phi^*(d_Y))\equiv \sum_{i,j\in X}\alpha_{ij} \cdot \phi^*(d_Y)(i,j) > Q^*\times \sum_{i,j\in X}\alpha_{ij} d_X(i,j).$$

The linear program for finding $d_X$ is given below.

\rule{0pt}{12pt}
\hrule height 0.8pt
\rule{0pt}{1pt}
\hrule height 0.4pt
\rule{0pt}{6pt}

\noindent \textbf{maximize:} 
$$\sum_{i,j\in X}\sum_{p,q \in Y}\alpha_{ij} \cdot \phi^*_{ipjq} d_X(p,q) - Q^*\times \sum_{i,j\in X}\alpha_{ij}d_X(i,j)$$

\noindent \textbf{subject to:} $d_X\in \calD_X$

\rule{0pt}{1pt}
\hrule height 0.4pt
\rule{0pt}{1pt}
\hrule height 0.8pt
\rule{0pt}{12pt}

\noindent If the maximum is greater than 0 for some $d^*_X$, then constraint~(\ref{eq:dY})
is violated for $d^*_Y={d^*_{X}}|_Y$ (the restriction of $d^*_X$ to $Y$), because
$$\minext_{Y\to X}(d^*_Y,\alpha) \leq \sum_{i,j\in X} \alpha_{ij} d^*_X(i,j).$$
If the maximum is 0 or negative, then all constraints~(\ref{constr:1}) are satisfied, simply because
$$\minext_{Y\to X}(d^*_Y,\alpha) = \min_{d_X:d_X\text{ is extension of  } d^*_Y}
\sum_{i,j\in X} \alpha_{ij} d_X(i,j).$$
\end{proof}

\section{Lipschitz Extendability}\label{sec:lip-extend}
In this section, we present exact bounds on the quality 
of cut and metric sparsifiers in terms of \textit{Lipschitz extendability constants}. We show that there exist
cut sparsifiers of quality $e_k(\ell_1, \ell_1)$ and
metric sparsifiers of quality  $e_k(\infty, \ell_\infty\oplus_1 \dots\oplus_1\ell_{\infty})$,
where $e_k(\ell_1, \ell_1)$ and $e_k(\infty, \ell_\infty\oplus_1 \dots\oplus_1\ell_{\infty})$
are the Lipschitz extendability constants (see below for the definitions).
We prove that these bounds are tight. 
Then we obtain a lower bound of $\Omega(\sqrt{\log k/\log\log k})$ for the quality
of the metric sparsifiers by proving a lower bound on $e_k(\infty, \ell_\infty\oplus_1 \dots\oplus_1\ell_{\infty})$.
In the first preprint of our paper, we also proved
the bound of $\Omega(\sqrt[4]{\log k/\log\log k})$ on $e_k(\ell_1,\ell_1)$.
After the preprint appeared on arXiv.org, Johnson and Schechtman notified us that a lower
bound of $\Omega(\sqrt{\log k}/\log\log k)$ on $e_k(\ell_1,\ell_1)$ follows from their joint 
work with Figiel~\cite{FJS}. With their permission, we present the proof of this lower bound 
in Section~\ref{sec:improved-l1} of the Appendix. This result implies a lower bound of $\Omega(\sqrt{\log k} / \log\log k)$ 
on the quality of cut sparsifiers.

On the positive side, we show that if 
a certain open problem in functional analysis posed 
by~\citeasnoun{Ball}
(see also \citeasnoun{LN}, and \citeasnoun{Randrianantoanina})
has a positive answer then $e_k(\ell_1, \ell_1) \leq 
\tilde O(\sqrt{\log k})$; and therefore there exist
$\tilde O(\sqrt{\log k})$-quality cut sparsifiers.
This is both an indication that the current upper bound
of $O(\log k/\log\log k)$ might not be optimal and
that improving lower bounds beyond of $\tilde O(\sqrt{\log k})$
will require solving a long standing open problem (negatively).
\begin{question}[~\citeasnoun{Ball}; see also \citeasnoun{LN} and \citeasnoun{Randrianantoanina}]
\label{qst:ball}
Is it true that $e_k(\ell_2, \ell_1)$ is bounded by a constant that does not depend on $k$?
\end{question}

Given two metric spaces $(X, d_X)$ and $(Y, d_Y)$, the Lipschitz extendability
constant $e_k(X, Y)$ is the infimum over all constants $K$ such that for every 
$k$ point subset $Z$ of $X$, every Lipschitz map $f:Z\to Y$ can be extended to
a map $\tilde f: X\to Y$ with $\|\tilde f\|_{Lip} \leq K \|f\|_{Lip}$.
We denote the supremum of $e_k(X,Y)$ over all separable metric spaces $X$ by
$e_k(\infty, Y)$. We refer the reader to \citeasnoun{LN}
for a background on the Lipschitz extension problem (see also
\citeasnoun{Kirszbraun}, \citeasnoun{McShane},
\citeasnoun{MarcusPisier}, \citeasnoun{JL},
\citeasnoun{Ball}, \citeasnoun{MendelNaor},
\citeasnoun{NPSS}). 
Throughout this section, $\ell_1$, $\ell_2$ and $\ell_\infty$ 
denote finite dimensional spaces of arbitrarily large dimension. 

In Section~\ref{subsec:qual-exten}, we establish the connection between
the quality of vertex sparsifiers and extendability constants. 
In Section~\ref{subsec:proj-const}, we prove lower bounds
on extendability constants $e_k(\infty,\ell_1)$ and $e_k(\ell_1, \ell_1)$,
which imply lower bounds on the quality of metric and cut
sparsifiers respectively. Finally, in Section~\ref{subsec:cond-upper-bound},
we show that if Question~\ref{qst:ball} (the open problem of Ball)
has a positive answer then there exist $\tilde O(\sqrt{\log k})$-quality cut sparsifiers.

\subsection{Quality of Sparsifiers and Extendability Constants}
\label{subsec:qual-exten}
Let $Q_k^{cut}$ be the minimum over all $Q$ such that there exists a $Q$-quality
\textit{cut} sparsifier for every graph $G=(V,\alpha)$ and every
subset $U\subset V$ of size $k$.
Similarly, let $Q_k^{metric}$ be the minimum over all $Q$ such that there exists
a $Q$-quality metric sparsifier for every graph $G=(V,\alpha)$
and every subset $U\subset V$ of size $k$.

\begin{theorem} There exist cut sparsifiers of quality $e_k(\ell_1,\ell_1)$ for subsets of size $k$. Moreover, this bound is tight. That is,
$$Q_k^{cut} = e_k(\ell_1,\ell_1).$$
\label{thm:extendcut}
\end{theorem}
\begin{proof}
Denote $Q = e_k(\ell_1,\ell_1)$.
First, we prove the existence of $Q$-quality cut sparsifiers.
We consider a graph $G = (V, \alpha)$ and a subset $U\subset V$ of size $k$. Recall that for every
cut $(S, U \setminus S)$ of $U$, the cost of the minimum cut extending $(S, U \setminus S)$
to $V$ is $\minext_{U\to V}(\delta_S, \alpha)$, where $\delta_S$ is the cut metric corresponding 
to the cut $(S, U \setminus S)$. Let $C = \{(\delta_S, \minext_{U\to V}(\delta_S, \alpha)) 
\in \calD_U \times \bbR: \delta_S \text{ is a cut metric}\}$ be the graph of the function
$\delta_S \mapsto \minext_{U\to V}(\delta_S, \alpha)$; and $\cal C$ be the 
convex cone generated by $C$ (i.e., let $\cal C$ be the cone over the convex closure of 
$C$).
Our goal is to construct a linear form $\beta$ (a cut sparsifier) with non-negative coefficients
such that  $x \leq \beta(d_U) \leq Q x$ for every $(d_U, x) \in \cal C$ and, in particular,
for every $(d_U, x) \in C$.
First we prove that for every $(d_1,x_1), (d_2, x_2) \in {\cal C}$ there exists
$\beta$ (with nonnegative coefficients) such that  
$x_1 \leq \beta(d_1)$ and $\beta(d_2) \leq Q x_2$.
Since these two inequalities are homogeneous, 
we may assume by rescaling $(d_2, x_2)$ that $Qx_2 = x_1$. We are going to show 
that for some $p$ and $q$ in $U$: $d_2(p,q) \leq d_1(p,q)$ and $d_1(p,q) \neq 0$. Then the linear 
form 
$$\beta(d_U) = \frac{x_1}{d_1(p,q)} d_U(p,q)$$
satisfies the required conditions:
$\beta(d_1) = x_1$; $\beta(d_2) = 
x_1d_2(p,q)/d_1(p,q) \leq x_1 = Qx_2$. 

Assume to the contrary that that for every $p$ and $q$,
$d_1(p,q) < d_2(p,q)$ or $d_1(p,q) = d_2(p,q) = 0$. 
Since $(d_t(p,q), x_t) \in {\cal C}$ for $t\in\{1,2\}$, by Carath\'eodory's theorem
$(d_t(p,q), x_t)$ is a convex combination of at most 
$\dim {\cal C} + 1 = \binom{k}{2} + 2$ points lying on the
extreme rays of $\cal C$. That is, there exists
a set of $m_t \leq \binom{k}{2} + 2$
positive weights $\mu_t^S$ such that  
$d_t = \sum_S \mu_t^S \delta_S$, where $\delta_S \in {\cal D}_U$ is the cut metric corresponding to the cut $(S, U \setminus S)$,
and $x_t = \sum_S \mu_t^S\minext_{U\to V}(\delta_S, \alpha)$.
We now define two maps $f_1: U \to \bbR^{m_1}$ and $f_2:V \to \bbR^{m_2}$.
Let $f_1(p) \in \bbR^{m_1}$ be a vector with one component $f_1^S(p)$ for each cut
$(S, U\setminus S)$ such that $\mu_1^S > 0$. Define $f_1^S(p) = \mu_1^S$ if $p \in S$; $f_2^S(p) = 0$, otherwise.
Similarly, let $f_2(i) \in \bbR^{m_2}$ be a vector with one component $f_2^S(i)$ for each cut
$(S, U\setminus S)$ such that $\mu_2^S > 0$. Let $(S^*, V\setminus S^*)$ be the minimum cut 
separating $S$ and $U\setminus S$ in $G$. Define $f_2^S(i)$ as follows: $f_2^S(i) = \mu_2^S$ if $i \in S^*$; $f_2^S(i) = 0$, otherwise.
%
%
Note that $\|f_1(p) - f_1(q)\|_1 = d_1(p,q)$ 
and $\|f_2(p) - f_2(q)\|_1 = d_2(p,q)$.
Consider a map  $g = f_1 f_2^{-1}$ from $f_2(U)$ to $f_1(U)$ (note that if $f_2(p) = f_2(q)$ then $d_2(p,q) = 0$,
therefore, $d_1(p,q) = 0$ and $f_1(p) = f_2(q)$; hence $g$ is well-defined).
For every $p$ and $q$ with $d_2(p,q)\neq 0$, 
$$
\|g(f_2(p)) - g(f_2(q))\|_1 = \|f_1(p) - f_1(q)\|_1 = 
d_1(p,q) < d_2(p,q) = \|f_2(p) - f_2(q)\|_1.
$$
That is, $g$ is a strictly contracting map. Therefore, there exists an extension of 
$g$ to a map $\tilde g: f_2(V) \to \bbR^{m_1}$ such that 
$$\|\tilde g(f_2(i)) - \tilde g(f_2(j))\|_1 < Q \|f_2(i) - f_2(j)\|_1
= Q d_2(i,j).$$
Denote the coordinate of $\tilde g(f_2(i))$ corresponding to the cut $(S,U\setminus S)$ 
by $\tilde g^S(f_2(i))$. 
Note that $\tilde g^S(f_2(p))/\mu_1^S = f_1^S(p)/\mu_1^S$ equals $1$ when $p\in S$ and $0$ when
$p\in U\setminus S$. Therefore, the metric $\delta^*_S(i,j) \equiv |\tilde g^S(f_2(i)) - \tilde g^S(f_2(j))|/\mu_1^S$  is an extension
of the metric $\delta_S(i,j)$ to $V$.
Hence,
$$\sum_{i,j\in V} \alpha_{ij}\delta^*_S(i,j)
\geq  \minext_{U\to V}(\delta_S, \alpha).$$
We have, 
\begin{align*}
x_1 &= \sum_{S} \mu_1^S \minext_{U\to V}(\delta_S, \alpha) \leq \sum_S \mu_1^S\sum_{i,j\in V} \alpha_{ij}\delta^*_S(i,j)
= \sum_S \sum_{i,j\in V} \alpha_{ij} |\tilde g^S(f_2(i)) - \tilde g^S(f_2(j))|\\
&= \sum_{i,j\in V} \alpha_{ij} \|\tilde g(f_2(i)) - \tilde g(f_2(j))\|_1
<\sum_{i,j\in V} Q \alpha_{ij} d_2(i,j) = Qx_2.
\end{align*}
We get a contradiction. We proved that for every $(d_1,x_1), (d_2, x_2) \in {\cal C}$ there exists
$\beta$ such that $x_1 \leq \beta(d_1)$ and $\beta(d_2) \leq Q x_2$.

Now we fix a point $(d_1, x_1) \in {\cal C}$ and consider
the set $\cal B$ of all linear functionals with nonnegative coefficients
$\beta$ such that  $x_1 \leq \beta(d_1)$. This is a convex closed set.
We just proved that for every $(d_2, x_2) \in \cal C$
there exists $\beta \in \cal B$ such that 
$Q x_2 - \beta(d_2) \geq 0$. Therefore, 
by the \citeasnoun{Neumann} minimax theorem,
there exist $\beta \in \cal B$ such that  for every 
$(d_2, x_2) \in \cal C$, $Q x_2 - \beta(d_2) \geq 0$.
Now we consider the set $\cal B'$ of all linear functionals $\beta$
with nonnegative coefficients such that  $Q x_2 - \beta(d_2) \geq 0$
for every $(d_2, x_2) \in \cal C$. Again, for every
$(d_1, x_1) \in \cal C$ there exists $\beta \in {\cal B}'$
such that  $\beta(d_1) - x_1 \geq 0$; therefore, by the minimax 
theorem there exists $\beta$ such that 
$x \leq \beta(d_U) \leq Q x$
for every $(d,x) \in \cal C$. We proved that there exists
a $Q$-quality cut sparsifier for $G$.

Now we prove that if for every graph $G = (V,\alpha)$ and a subset $U\subset V$
of size $k$ there exists a cut sparsifier of size $Q$ (for some $Q$) then $e_k(\ell_1, \ell_1) \leq Q$.
Let $U \subset \ell_1$ be a set of points of size $k$ and $f:U \to \ell_1$ be 
a 1-Lipschitz map. By a standard compactness argument (Theorem~\ref{thm:compactness}), it suffices to show how to extend $f$ to a $Q$-Lipschitz map
$\tilde f:V \to \ell_1$ for every finite set $V$: $U\subset V\subset \ell_1$. 
First, we assume that $f$ maps $U$ to the vertices of a rectangular box 
$\{0,a_1\}\times \{0,a_2\} \times \dots \{0,a_r\}$. We consider a
graph $G = (V,\alpha)$ on $V$ with nonnegative edge weights $\alpha_{ij}$.
Let $(U, \beta)$ be the optimal cut sparsifier of $G$. 
Denote $d_1(p,q) = \|p-q\|_1$ and $d_2(p,q) = \|f(p)-f(q)\|_1$.
Since $f$ is 1-Lipschitz, $d_1(p,q) \geq d_2(p,q)$.

Let $S_i = \{p \in U: f_i(p) = 0\}$  (for $1\leq i\leq r$). Let $S_i^*$ be the minimum cut separating
$S_i$ and $U \setminus S_i$ in $G$. By the definition of the cut sparsifier, the cost 
of this cut is at most $\beta(\delta_{S_i})$. Define an extension $\tilde f$ of $f$ by
$\tilde f_i(v) = 0$ if $v \in S_i^*$ and $\tilde f_i(v) = a_i$ otherwise.
Clearly, $\tilde f$ is an extension of $f$. We compute the ``cost'' of $\tilde f$:
$$
\sum_{u,v\in V}\alpha_{uv} \|\tilde f(u) - \tilde f(v)\|_1 
= 
\sum_{i=1}^r \sum_{u,v\in V}\alpha_{uv} |\tilde f_i(u) - \tilde f_i(v)|\leq
\sum_{i=1}^r \beta(a_i \delta_{S_i}) = \beta(d_2) \leq \beta(d_1).$$
(in the last inequality we use that $d_1(p,q) \geq d_2(p,q)$ for $p,q \in U$
and that coefficients of $\beta$ are nonnegative). 
On the other hand, we have
$$\sum_{u,v\in V}\alpha_{uv} \|u - v\|_1 \geq \minext_{U\to V}(d_1, \alpha) 
\geq \beta(d_1)/Q.$$
We therefore showed that for every set of nonnegative weights $\alpha$
there exists an extension $\tilde f$ of $f$ such that 
\begin{equation}
\sum_{u,v\in V}\alpha_{uv} \|\tilde f(u) - \tilde f(v)\|_1 
\leq Q\sum_{u,v\in V}\alpha_{uv} \|u - v\|_1.
\label{inq:weightedlip}
\end{equation}
Note that the set of all extensions of $f$ is a closed convex set;
and $\|f(u) - f(v)\|_1$ is a convex function of $f$:
$$\|(f_1+f_2)(u) - (f_1 + f_2)(v)\|_1 \leq \|f_1(u) - f_1(v)\|_1 +
\|f_2(u) - f_2(v)\|_1.$$
Therefore, by the \citeasnoun{Sion} minimax theorem there exists an extension $\tilde f$
such that  inequality~(\ref{inq:weightedlip}) holds for every nonnegative $\alpha_{ij}$.
In particular, when $\alpha_{uv} =1$ and all other $\alpha_{u'v'} = 0$,
we get
$$\|\tilde f(u) - \tilde f(v)\|_1 \leq Q\|u - v\|_1.$$
That is, $\tilde f$ is $Q$-Lipschitz.

Finally, we consider the general case when the image of $f$ is not necessarily
a subset of $\{0,a_1\}\times \{0,a_2\} \times \dots \{0,a_r\}$. Informally, we are going to replace $f$
with an ``equivalent map'' $g$ that maps $U$ to vertices of a rectangular box, then apply our result to $g$,
obtain a $Q$-Lipschitz extension $\tilde g$ of $f$, and finally replace $\tilde g$ with an extension $\tilde f$ of
$f$. 

Let $f_i(p)$ be the $i$-th coordinate of $f(p)$. Let $b_1, \dots, b_{s_i}$ be the set of values
of $f_i(p)$ (for $p\in U$). Define map $\psi_i: \{b_1,\dots, b_{s_i}\} \to \bbR^{s_i}$ as
$\psi_i(b_j) = (b_1, b_2 - b_1,\dots, b_j - b_{j-1}, 0, \dots, 0)$. The map $\psi_i$ is 
an isometric embedding of $\{b_j\}$ into $(\bbR^{s_i},\|\cdot\|_1)$. 
Define map $\phi_i$ from $(\bbR^{s_i},\|\cdot\|_1)$ to $\mathbb R$ as
$\phi_i(x) = \sum_{t=1}^{s_i} x_t$. Then $\phi_i$ is 1-Lipschitz
and $\phi_i(\psi_i(b_j)) = b_j$. 
Now let
\begin{align*}
g (p) &= \psi_1(f_1(p)) \oplus  \psi_2(f_2(p)) \oplus \dots \oplus 
\psi_r(f_r(p)) \in \bigoplus_{i=1}^r \bbR^{s_i},\\
\phi(y_1 \oplus \dots\oplus y_r) &= \phi_1(y_1) \oplus \phi_2(y_2) \oplus \dots \oplus
\phi_r(y_r) \in \ell_1^r
\end{align*}
(where $r$ is the number of coordinates of $f$).
Since maps $\psi_i$ are isometries and $f$ is 1-Lipschitz, $g$ is 1-Lipschitz
as well. Moreover, the image of $g$ is a subset of vertices of a box.
Therefore, we can apply our extension result to it. We obtain a $Q$-Lipschitz
map $\tilde g:V \to \bigoplus_{i=1}^r \bbR^{s_i}$. 
$$
\xymatrix @C=4.7pc {
U \ar@{^{(}->}[d]^{\subset} \ar[r]^f \ar@/^1.8pc/[rr]^g &
f(U)
\ar@{^{(}->}[d]^{\subset} \ar[r]^{\psi_1\oplus\dots\oplus\psi_r} &
**[r]
\bigoplus_{i=1}^r \bbR^{s_i}\ar@{=}[d]\\
V \ar[r]^{\tilde f} \ar@/_1.4pc/[rr]^{\tilde g} &\ell_1^r \ar@{<-}[r]^{\phi = \phi_1\oplus\dots\oplus\phi_r}  &
**[r]
\bigoplus_{i=1}^r \bbR^{s_i}
}
$$
Note also that $\phi$ is 1-Lipschitz and $\phi(g(p)) = f(p)$.
Finally, we define $\tilde f(u) = \phi(\tilde g(u))$. We have $\|\tilde f\|_{Lip} \leq
\|\tilde g\|_{Lip} \|\phi\|_{Lip} \leq Q$.  This concludes the proof. 
\end{proof}


\begin{theorem}
\label{thm:extendmetric}
There exist metric sparsifiers of quality $e_k(\infty,\ell_{\infty}\oplus_1 \dots \oplus_1 \ell_{\infty})$
for subsets of size $k$ and this bound is tight.
Since $\ell_1$ is a Lipschitz retract of
$\ell_{\infty}\oplus_1 \dots \oplus_1 \ell_{\infty}$ (the retraction projects 
each summand $L_i = \ell_{\infty}$ to the first coordinate of $L_i$),
$e_k(\infty,\ell_{\infty}\oplus_1 \dots \oplus_1 \ell_{\infty}) 
\geq e_k(\infty,\ell_1)$. Therefore, the quality of metric sparsifiers
is at least $e_k(\infty,\ell_1)$ for some graphs. In other words,
$$Q_k^{metric} = e_k(\infty,\ell_{\infty}\oplus_1 \dots \oplus_1 \ell_{\infty}) \geq
e_k(\infty,\ell_1).$$
\end{theorem}
\begin{proof}
Let $Q = e_k(\infty, \ell_{\infty}\oplus_1 \dots \oplus_1 \ell_{\infty})$.
We denote the norm of a vector $v\in \ell_{\infty}\oplus_1 \dots \oplus_1 \ell_{\infty}$
by $\|v\|\equiv \|v\|_{\ell_{\infty}\oplus_1 \dots \oplus_1 \ell_{\infty}}$.
First, we construct a $Q$-quality metric sparsifier for a given 
graph $G = (V, \alpha)$ and $U\subset V$ of size $k$.

Let $C = \{(d_U, \minext_{U\to V}(d_U, \alpha)): d_U \in {\cal D}_U \}$
and $\cal C$ be the convex hull of $C$.
We construct a linear form $\beta$ (a metric sparsifier) with non-negative coefficients
such that $x \leq \beta(d_U) \leq Q x$ for every $(d_U, x) \in \cal C$.

The proof follows the lines of Theorem~\ref{thm:extendcut}.
The only piece of the proof that we need to modify slightly is the
proof that the following is impossible: 
for some $(d_1,x_1)$ and $(d_2, x_2)$ in $\cal C$, $x_1 = Qx_2$ and
for all $p,q\in U$ either $d_1(p,q) < d_2(p,q)$ or $d_1(p,q) = d_2(p,q) = 0$.
Assume the contrary. We represent $(d_1,x)$ as a convex combination of 
points $(d_1^i, x_1^i)$ in $C$ (by Carath\'eodory's theorem). 
Let $f_i$ be an isometric embedding of the metric space $(U,d_1^i)$ into
$\ell_{\infty}$. Then $f\equiv \oplus_i f_i$ is an isometric embedding of $(U, d_1)$
into $\ell_{\infty}\oplus_1 \dots \oplus_1 \ell_{\infty}$.
Let $d_2^*$ be the minimum extension of $d_2$ to $V$.
Note that $f$ is a strictly contracting map from $(U, d_2)$ to $\ell_{\infty}\oplus_1 \dots \oplus_1 \ell_{\infty}$: 
$$\|f(p) - f(q)\|_{\infty} = \sum_i \|f_i(p) - f_i(q)\|_{\infty} = 
\sum_i d_1^i(p,q) = d_1(p,q) < d_2(p,q),$$ 
for all $p, q\in U$ such that $d_2(p,q) > 0$.
Therefore, there exists a Lipschitz extension of $f: (U, d_2) \to \ell_{\infty}\oplus_1 \dots \oplus_1 \ell_{\infty}$ to $\tilde f : (V, d_2^*) \to \ell_{\infty}\oplus_1 \dots \oplus_1 \ell_{\infty}$
with $\|\tilde f\|_{Lip} < Q$. Let $\tilde f_i:V \to \ell_{\infty}$ be the projection of $f$ to the $i$-th summand.
Let $\tilde d_1^i(x,y) = \|\tilde f_i(x) - \tilde f_i(y)\|_{\infty}$ be the metric induced by $\tilde f_i$ on $G$.
Let 
$$\tilde d_1(x,y) = \|\tilde f(x) - \tilde f(y)\|_{\infty} = \sum_i \|\tilde f_i(x) - \tilde f_i(y)\|_{\infty} = \sum_i \tilde d_1^i(x,y)$$
be the metric induced by $\tilde f$ on $G$.
Since $\tilde f_i(p) = f_i(p)$ for all $p\in U$, metric $\tilde d_1^i$ is an extension of
$d_1^i$ to $V$.  Thus $\alpha(\tilde d_1^i) \geq \minext_{U\to V}(d_1^i, \alpha) = x_1^i$.
Therefore, $\alpha(\tilde d_1) = \alpha(\sum \tilde d_1^i) \geq \sum_i x_1^i = x_1$.
Since $\|\tilde f\|_{Lip} < Q$, 
$\tilde d_1(x,y) = \|\tilde f(x) - \tilde f(y)\|_{\infty} < Q d_2^*(x, y)$
(for every $x,y\in V$ such that $d_2^*(x,y) > 0$).
We have,
$$\alpha(\tilde d_1) < \alpha(Q d_2^*) = Q \minext_{U\to V}(d_2, \alpha) \leq Qx_2= x_1.$$
We get a contradiction.


Now we prove that if for every graph $G = (V,\alpha)$ and a subset $U\subset V$
of size $k$ there exists a metric sparsifier of size $Q$ (for some $Q$) then 
$e(\infty, \ell_{\infty}\oplus_1 \dots \oplus_1 \ell_{\infty}) \leq Q$.
Let $(V, d_V)$ be an arbitrary metric space; and $U \subset V$ be a subset of size $k$.
Let $f:(U, d_V|_U) \to \ell_{\infty}\oplus_1 \dots \oplus_1 \ell_{\infty}$ be a 1-Lipschitz map.
We will show how to extend $f$ to a $Q$-Lipschitz map $\tilde f:(V, d_V) \to \ell_{\infty}\oplus_1 \dots \oplus_1 \ell_{\infty}$. 
We consider graph $G = (V, \alpha)$ with nonnegative edge weights and
a $Q$-quality metric sparsifier $\beta$.

Let $f_i:U \to \ell_{\infty}$ be the projection of $f$ onto its $i$-th summand.
Map $f_i$ induces metric $d^i(p,q) = \|f_i(p) - f_i(q)\|$ on $U$. 
Let $\tilde d^i$ be the minimum metric extension of $d^i$ to $V$; let $\tilde d_*(x,y) = \sum_i \tilde d^i(x, y)$.
Note that since $f$ is 1-Lipschitz 
$$\tilde d_*(p,q) = \sum_i \tilde d^i(p,q) = \sum_i \|f_i(p) - f_i(q)\| = \|f(p) - f(q)\| \leq d_V(p,q)$$
for $p,q\in U$.
Therefore,
$$\alpha(\tilde d_*) = \sum_i \alpha(\tilde d^i) \leq \sum_i \beta(d^i) = \beta(\tilde d_*|_U) \leq \beta(\left.d_V\right|_U) \leq Q \alpha(d_V)$$
(we use that all coefficients of $\beta$ are nonnegative).

Each map $f_i$ is an isometric embedding of $(U, d^i)$ to $\ell_{\infty}$ (by the definition of $d^i$).
Using the McShane extension theorem\footnote{The McShane extension theorem states that $e_k(M, \bbR) =1$ for every metric space $M$.} \cite{McShane}, we extend each $f_i$ a 1-Lipschitz map $\tilde f_i$ from $(V, \tilde d^i)$ to $\ell_{\infty}$.
Finally, we let $\tilde f = \oplus_i \tilde f_i$. Since each $\tilde f_i$ is an extension of $f_i$,
$\tilde f$ is an extension of $f$. For every $x,y\in V$, we have $\|\tilde f(x) - \tilde f(y)\| 
= \sum_i \|\tilde f_i(x) - \tilde f_i(y)\| = \tilde d_*(x, y)$. Therefore,
$$\sum_{x,y\in V} \alpha_{xy} \|\tilde f(x) - \tilde f(y)\|  = \alpha(\tilde d_*) \leq Q\alpha(d_V) =
Q\times \sum_{x,y\in V} \alpha_{xy} d_V(x,y).$$
We showed that for every set of nonnegative weights $\alpha$
there exists an extension $f$ such that the inequality above holds.
Therefore, by the minimax theorem there exists an extension $\tilde f$
such that this inequality holds for every nonnegative $\alpha_{xy}$.
In particular, when $\alpha_{xy} =1$ and all other $\alpha_{x'y'} = 0$,
we get
$$\|\tilde f(x) - \tilde f(y)\| \leq Q d_V(x,y).$$
That is, $\tilde f$ is $Q$-Lipschitz.
\end{proof}
\begin{remark} We proved in Theorem~\ref{thm:extendcut}
that $Q_k^{cut} = e_k(\ell_1^M, \ell_1^N)$ for $\binom{k}{2} + 2 \leq M, N <
\infty$; by a simple compactness argument the equality also holds when either one or both
of $M$ and $N$ are equal to infinity. Similarly,  we proved in Theorem~\ref{thm:extendmetric}
that $Q_k^{metric} = e_k(\infty, \underbrace{\ell_\infty^M \oplus_1 \dots \oplus_1
\ell_\infty^M}_N)$ for $k-1 \leq M < \infty$ and $\binom{k}{2} + 2 \leq N <
\infty$; this equality also holds when either one or both
of $M$ and $N$ are equal to infinity. (We will not use use this observation.)
\end{remark}

\subsection{Lower Bounds and Projection Constants}\label{subsec:proj-const}
We now prove lower bounds on the quality of metric and cut sparsifiers.
We will need several definitions from analysis.  
The operator norm of a linear operator $T$ from a normed space
$U$ to a normed space $V$ is 
$\|T\| \equiv \|T\|_{U\to V} = \sup_{u\neq 0} \|Tu\|_V/\|u\|_U$.
The Banach--Mazur distance between two normed
spaces $U$ and $V$ is 
$$d_{BM}(U,V) = \inf\{\|T\|_{U\to V} \|T^{-1}\|_{V\to U}: 
T \text{ is a linear operator from } U \text{ to } V\}.$$
We say that two Banach spaces are $C$-isomorphic if the Banach--Mazur
distance between them is at most $C$; two Banach spaces are isomorphic
if the Banach--Mazur distance between them is finite.
A linear operator $P$ from a Banach space $V$ to a subspace $L\subset V$ is 
a projection if the restriction of $P$ to $L$ is the identity operator
on $L$ (i.e., $P|_L = I_L$).

Given  a Banach space $V$ and subspace $L \subset V$, we define the relative 
projection constant $\lambda(L, V)$ as:
$\lambda(L, V) = \inf \{\|P\|: P \text{ is a linear projection from } V \text{ to }L\}$.

\begin{theorem}\label{thm:boundQmetric}
$$Q_k^{metric} = \Omega(\sqrt{\log k / \log\log k}).$$
\end{theorem}
\begin{proof}
To establish the theorem, we prove lower bounds for $e_k(\ell_\infty, \ell_1)$.
Our proof is a modification of the proof of \citeasnoun{JL}
that $e_k(\ell_1,\ell_2) = \Omega(\sqrt{\log k / \log\log k})$. 
Johnson and Lindenstrauss showed that for every space $V$ and subspace $L\subset V$
of dimension $d = \lfloor {c\log k}/{\log\log k}\rfloor$, $e_k(V, L) = \Omega(\lambda(L, V))$ 
(\citeasnoun{JL}, see Appendix~\ref{sec:JL-overview}, Theorem~\ref{thm:JL}, for a sketch of the proof).

Our result follows from the lower bound of \citeasnoun{Grunbaum}: for a certain
isometric embedding of $\ell_1^d$ into $\ell_\infty^N$,
$\lambda(\ell_1^d, \ell_\infty^N) = \Theta(\sqrt{d})$ (for large enough $N$).
Therefore, $e_k(\ell_\infty^N, \ell_1^d) = \Omega(\sqrt{\log k / \log\log k})$.
\end{proof}

We now prove a lower bound on $Q_k^{cut}$. 
Note that the argument from Theorem~\ref{thm:boundQmetric} shows that
$Q_k^{cut} = e_k(\ell_1^d, \ell_1^N) = \Omega(\lambda(L, \ell_1^N))$,
where $L$ is a subspace of $\ell_1^N$ isomorphic to $\ell_1^d$. \citeasnoun{Bourgain} proved that 
there is a non-complemented subspace isomorphic to $\ell_1^{\infty}$ in $L_1$.
This implies that $\lambda(L, \ell_\infty^N)$ (for some $L$) and, therefore, $Q_k^{cut}$ are unbounded.
However, quantitatively Bourgain's result gives a very weak bound of (roughly) $\log\log\log k$.
It is not known how to improve Bourgain's bound.
So instead we present an explicit family of non-$\ell_1$ subspaces $\{L\}$ of $\ell_1$ with
$\lambda(L, \ell_1) = \Theta(\sqrt{\dim L})$
and $d_{BM}(L, \ell_1^{\dim L}) =O(\sqrt[4]{\dim L})$.  

\begin{theorem}\label{thm:boundQcut}
$$Q_k^{cut} \geq \Omega(\sqrt[4]{\log k/\log\log k}).$$
\end{theorem}

We shall construct a $d$ dimensional subspace $L$ of $\ell^N_1$, with the
projection constant $\lambda(L, \ell_1) \geq \Omega(\sqrt{d})$ and
with Banach--Mazur distance $d(L, \ell_1^d) \leq O(\sqrt[4]{d})$.
By Theorem~\ref{thm:JL} (as in Theorem~\ref{thm:boundQmetric}),
$e_k(\ell_1,L) \geq \Omega(\sqrt{d})$
for $d = \lfloor c \log k/\log\log k\rfloor$. 
The following lemma then implies that $e_k(\ell_1, \ell^d_1) \geq \Omega(\sqrt[4]{d})$.

\begin{lemma}\label{lem:ineq-extension-bm}
For every metric space $X$ and finite dimensional normed spaces $U$ and $V$,
$$e_k(X,U) \leq e_k(X,V) d_{BM}(U,V).$$
\end{lemma}
\begin{proof}
 Let $T:U \to V$ be a linear operator with $\|T\|\|T^{-1}\| = d_{BM}(U,V)$.
Consider a $k$-point subset $Z\subset X$ and a Lipschitz map $f:Z\to U$.
Then $g = Tf$ is a Lipschitz map from $Z$ to $V$. Let $\tilde g$ be an extension
of $g$ to $X$ with $\|\tilde g\|_{Lip} \leq e_k(X,V)\|g\|_{Lip}$. Then 
$\tilde f = T^{-1} \tilde g$ is an extension of $f$ and
\begin{align*}
\|\tilde f\|_{Lip} &\leq \|T^{-1}\| \|\tilde g\|_{Lip}\leq 
\|T^{-1}\| \cdot e_k(X,V)\cdot \|g\|_{Lip}\\
&\leq  \|T^{-1}\| \cdot e_k(X,V) \cdot \|T\| \|f\|_{Lip} 
= e_k(X,V) d_{BM}(U,V) \|f\|_{Lip}.
\end{align*}
\end{proof}

\begin{proof}[Proof of Theorem~\ref{thm:boundQcut}]
Fix numbers  $m>0$ and $d=m^2$. Let $\calS\subset \bbR^d$ be the set
of all vectors in $\{-1,0,1\}^d$ having exactly $m$ nonzero coordinates.
Let $f_1,\dots, f_d$ be functions from $\calS$ to $\bbR$ defined as $f_i(S) = S_i$ 
($S_i$ is the $i$-th coordinate of $S$). These functions belong to the space 
$V = L_1(\calS,\mu)$ (where $\mu$ is the counting measure on $\calS$). 
The space $V$ is equipped with the $L_1$ norm
$$\|f\|_1 = \sum_{S\in \calS} |f(S)|;$$
and the inner product $$\langle f, g\rangle = \sum_{S\in \calS} f(S)g(S).$$
The set of indicator functions $\{e_S\}_{S\in\calS}$
$$e_S(A) = \begin{cases}1,&\text{if } A=S;\\0,&\text{otherwise}\end{cases}$$
is the standard basis in $V$.

Let $L\subset V$ be the subspace spanned by $f_1,\dots, f_d$. We prove that
the norm of the orthogonal projection operator $P^{\perp}:V\to L$ is at least $\Omega(\sqrt{d})$ and then using symmetrization show 
that $P^{\perp}$ has the smallest norm among
all linear projections. This approach is analogues to the approach 
of~\citeasnoun{Grunbaum}.

All functions $f_i$ are orthogonal and $\|f_i\|^2_2 = |\calS|/m$ (since for a random $S\in\calS$, 
$\Prob{f_i(S)\in \{\pm 1\}} = 1/m$). We find the projection of an arbitrary basis vector $e_A$ (where $A\in \calS$) on $L$,
\begin{eqnarray*}
P^{\perp}(e_A) &=&
\sum_{i=1}^d \frac{\langle e_A, f_i\rangle}{\|f_i\|^2}f_i 
=
\sum_{i=1}^d \sum_{B\in \calS}\frac{\langle e_A, 
f_i\rangle}{\|f_i\|^2}\langle f_i,e_B\rangle e_B\\
&=& 
\frac{m}{|\calS|}\sum_{B\in \calS} 
\left(\sum_{i=1}^d \langle e_A, f_i\rangle\langle f_i,e_B\rangle\right) e_B.
\end{eqnarray*}
Hence,
\begin{equation}\label{eq:perpnorm}
\|P^{\perp}(e_A)\|_1 = 
\frac{m}{|\calS|}\sum_{B\in \calS} 
\left|\sum_{i=1}^d \langle e_A, f_i\rangle\langle f_i,e_B\rangle\right|.
\end{equation}
Notice, that 
$$\sum_{i=1}^d \langle e_A, f_i\rangle\langle f_i,e_B\rangle = \sum_{i=1}^d A_i B_i
= \langle A, B\rangle.$$ 
For a fixed $A\in \calS$ and a random (uniformly distributed) $B\in\calS$ the probability that $A$ and $B$ 
overlap by exactly one nonzero coordinate (and thus $|\langle A, B \rangle| =1$) is
at least $1/e$. Therefore (from~(\ref{eq:perpnorm})),
$$\|P^{\perp}(e_A)\|_1 \geq \Omega(m) = \Omega(\sqrt{d}),$$
and $\|P^{\perp}\|\geq \|P^{\perp}(e_A)\|_1/\|e_A\|_1 \geq \Omega(\sqrt{d})$.

We now consider an arbitrary linear projection $P:L\to V$. We shall prove that
$$\sum_{A\in \calS} \|P(e_A)\|_1 - \|P^{\perp}(e_A)\|_1 \geq 0,$$
and hence for some $e_A$, $\|P(e_A)\|_1 \geq \|P^{\perp}(e_A)\|_1\geq \Omega(\sqrt{d})$.
Let $\sigma_{AB} = \sgn(\langle P^{\perp}(e_A), e_B\rangle) =
\sgn(\langle A, B \rangle)$. 
Then,
$$\|P^{\perp}(e_A)\|_1 = \sum_{B\in\calS}|\langle P^{\perp}(e_A), e_B\rangle| =
\sum_{B\in\calS}\sigma_{AB}\langle P^{\perp}(e_A), e_B\rangle,$$
and, since $\sigma_{AB}\in [-1,1]$,
$$\|P(e_A)\|_1 = \sum_{B\in\calS}|\langle P(e_A), e_B\rangle| \geq
\sum_{B\in\calS}\sigma_{AB}\langle P(e_A), e_B\rangle.$$
Therefore,
$$\sum_{A\in \calS} \|P(e_A)\|_1 - \|P^{\perp}(e_A)\|_1
\geq \sum_{A\in \calS}\sum_{B\in\calS}\sigma_{AB}\langle P(e_A) - P^{\perp}(e_A), e_B\rangle.$$
Represent operator $P$ as the sum
$$P(g) = P^{\perp}(g) + \sum_{i=1}^d \psi_i(g) f_i,$$
where $\psi_i$ are linear functionals\footnote{The explicit expression for $\psi_i$ is as follows $\psi_i (g) = \langle P(g) - P^{\perp}(g), f_i \rangle/\|f_i\|^2$.} with $\ker \psi_i \supset L$.
We get
\begin{eqnarray*}
\sum_{A\in \calS}\sum_{B\in\calS}\sigma_{AB}\langle P(e_A) - P^{\perp}(e_A), e_B\rangle
&=& 
\sum_{A\in \calS}\sum_{B\in\calS}\sigma_{AB}\langle \sum_{i=1}^d \psi_i(e_A) f_i, e_B\rangle\\
&=& 
\sum_{i=1}^d
\psi_i\left(\sum_{A\in \calS} \sum_{B\in\calS} \sigma_{AB} \langle  e_B, f_i\rangle e_A\right).
\end{eqnarray*}
We now want to show that each vector 
$$
g_i = \sum_{A\in \calS} \sum_{B\in\calS} \sigma_{AB} \langle  e_B, f_i\rangle e_A
$$
is collinear with $f_i$, and thus $g_i\in L\subset \ker \psi_i$ and $\psi_i(g_i)=0$.
We need to compute $g_i(S)$ for every $S\in \calS$,
$$g_i(S) = \sum_{A\in \calS} \sum_{B\in\calS} 
\sigma_{AB} \langle  e_B, f_i\rangle e_A(S)=
\sum_{B\in\calS} \sigma_{SB} B_i,$$
we used that $e_A(S) = 1$ if $A=S$, and $e_A(S) = 0$ otherwise.
We consider a group $H \cong \mathbb{S}_d \ltimes \mathbb{Z}_2^d$ of symmetries
of $\calS$. The elements of $H$ are pairs $h = (\pi,\delta)$, where each $\pi \in {\mathbb S}_d$ is 
a permutation on $\{1,\dots, d\}$, and each $\delta\in \{-1,1\}^d$.
The group acts on $\calS$ as follows: it first permutes the coordinates of every 
vector $S$ according to $\pi$ and then changes the signs of the $j$-th coordinate
if $\delta_j = -1$ i.e.,
$$h: S = (S_1,\dots, S_d)\mapsto hS = (\delta_1 S_{\pi^{-1}(1)},\dots, \delta_d S_{\pi^{-1}(d)}).$$
The action of $G$ preserves the inner product between $A,B\in \calS$ i.e.,
$\langle hA, hB\rangle  = \langle A,B\rangle$ and thus $\sigma_{(hA)(hB)}=\sigma_{AB}$.
It is also transitive. Moreover, for every $S, S'\in \calS$, if $S_i = S'_i$,
then there exists $h\in G$ that maps $S$ to $S'$, but does not change 
the $i$-th coordinate (i.e., $\pi(i) = i$ and $\delta_i = 1$). 
Hence, if $S_i=S'_i$, then for some $h$
$$g_i(S') = g_i(hS) = 
\sum_{B\in\calS} \sigma_{(hS)B} B_i=
\sum_{B\in\calS} \sigma_{(hS)(hB)} (hB)_i
= \sum_{B\in\calS} \sigma_{SB} (hB)_i
= \sum_{B\in\calS} \sigma_{SB} B_i = g_i(S)
.$$
On the other hand,  $g_i(S) = -g_i(-S)$. Thus, if $S_i=-S'_i$, then 
$g_i(S) = -g_i(S')$. Therefore, 
$g_i(S) = \lambda S_i$ for some $\lambda$, and $g_i = \lambda f_i$.
This finishes the prove that $\|P\|\geq \Omega(\sqrt{d})$.

We now estimate the Banach--Mazur distance from $\ell_1^d$ to $L$.

\begin{lemma}
We say that a basis $f_1, \dots, f_d$ of a normed space $(L,\|\cdot\|_L)$ is symmetric if 
the norm of vectors in $L$  does not depend on the order and signs of coordinates in this basis:
$$\Bigl\|\sum_{i=1}^d c_i f_i\Bigr\|_L = \Bigl\|\sum_{i=1}^d\delta_i c_{\pi(i)} f_i\Bigr\|_L,$$
for every $c_1,\dots, c_d\in \bbR$, $\delta_1,\dots,\delta_d\in\{\pm1\}$ and $\pi \in {\mathbb S}_d$.

Let $f_1,\dots, f_d$ be a symmetric basis. Then 
$$d_{BM}(L, \ell_1^d) \leq \frac{d\|f_1\|_L}{\|f_1 + \dots + f_d\|_L}.$$ 
\end{lemma}
\begin{proof}
Denote by $\eta_1,\dots \eta_d$ the standard basis of $\ell_1^d$. Define a linear operator 
$T:\ell_1^d \to L$ as $T(\eta_i) = f_i$. Then $d_{BM}(L,\ell_1^d) \leq \|T\|\cdot\|T^{-1}\|$.
We have,
\begin{eqnarray*}
\|T\| &=& \max_{c\in \ell_1^d: \|c\|_1=1} 
\|T(c_1\eta_1 + \dots + c_d\eta_d)\|_L \leq 
\max_{c\in \ell_1^d: \|c\|_1=1} (\|T(c_1\eta_1)\|_L + \dots + \|T(c_d\eta_d)\|_L)\\
&=& \max_{i} \|T(\eta_i)\|_L = \max_{i} \|f_i\|_L = \|f_1\|_L. 
\end{eqnarray*}
On the other hand,
\begin{eqnarray*}
(\|T^{-1}\|)^{-1} &=& 
\min_{c\in \ell_1^d: \|c\|_1=1} \|T^{-1}(c_1\eta_1 + \dots + c_d\eta_d)\|_L =
\min_{c\in \ell_1^d: \|c\|_1=1} \|c_1 f_1 + \dots + c_d f_d\|_L.
\end{eqnarray*}
Since the basis $f_1,\dots, f_d$ is symmetric, we can assume that all $c_i\geq 0$.
We have,
$$\Bigl\|\sum_{i=1}^d c_i f_i\Bigr\|_L = \Exp_{\pi\in {\mathbb S}_d}
\Bigl\|\sum_{i=1}^d c_{\pi(i)} f_i\Bigr\|_L \geq \Bigl\|\Exp_{\pi\in {\mathbb S}_d}
\sum_{i=1}^d c_{\pi(i)} f_i\Bigr\|_L = \Bigl\|\frac{1}{d}\sum_{i=1}^d f_i\Bigr\|_L.$$
\end{proof}
We apply this lemma to the space $L$ and basis $f_1,\dots, f_d$.
Note that $\|f_i\|_1 = |\calS|/m$ and
$$\|f_1 + \dots + f_d\|_1 = \sum_{S\in \calS}\Bigl|\sum_{i=1}^d S_i\Bigr|.$$
Pick a random $S\in \calS$. Its $m$ nonzero coordinates distributed
according to the Bernoulli distribution, thus $\Bigl|\sum_i S_i\Bigr|$ equals in
expectation $\Omega(\sqrt{m})$ and therefore the Banach--Mazur distance between
$\ell_1^d$ and $L$ equals 
$$d_{BM}(L,\ell_1^d) 
=  O\left(d \times \frac{|\calS|}{m} \times \frac{1}{\sqrt{m}|\calS|}\right)
= O(\sqrt[4]{d}).$$
\end{proof}

\subsection{Conditional Upper Bound and Open Question of Ball}\label{subsec:cond-upper-bound}
We show that if Question~\ref{qst:ball} (see page \pageref{qst:ball}) has a positive answer then there exist $\tilde O(\sqrt{\log k})$-quality cut sparsifiers.

\begin{theorem}
$$Q_k^{cut} = e_k(\ell_1, \ell_1) \leq O(e(\ell_2, \ell_1) \sqrt{\log k} \log\log k).$$
\end{theorem}
\begin{proof}
We show how to extend a map $f$ that maps a $k$-point subset $U$ of $\ell_1$ to $\ell_1$ to a map $\tilde f:\ell_1\to \ell_1$
via factorization through $\ell_2$. In our proof, we use 
a low distortion Fr\'echet embedding of a subset of $\ell_1$ into $\ell_2$
constructed by \citeasnoun{ALN}:
\begin{theorem}[\citeasnoun{ALN}, Theorem 1.1]\label{thm:ALN}
Let $(U, d)$ be a $k$-point subspace 
of $\ell_1$. Then there exists a probability measure $\mu$
over random non-empty subsets $A\subset U$ such that  for every $x,y \in U$
$$\Exp_\mu[|d(x,A) - d(y,A)|^2]^{1/2} = 
\Omega\left(\frac{d(x,y)}{\sqrt{\log k} \log\log k}\right).$$
\end{theorem}
We apply this theorem to the set $U$ with $d(x,y) = \|x-y\|_1$. We get a probability distribution $\mu$ of sets $A$.
Let $g$ be the map that maps each $x\in \ell_1$ to the random variable $d(x,A)$ in $L_2(\mu)$.
Since for every $x$ and $y$ in $\ell_1$,
$\Exp_\mu[|d(x,A) - d(y,A)|^2]^{1/2} \leq \Exp_\mu[\|x-y\|_1^2]^{1/2} = \|x-y\|_1$,
the map $g$ is a 1-Lipschitz map from $\ell_1$ to $L_2(\mu)$. On the other hand, 
Theorem~\ref{thm:ALN} guarantees that
the Lipschitz constant of $g^{-1}$ restricted to $g(U)$ is at most 
$O(\sqrt{\log k}\log\log k)$.

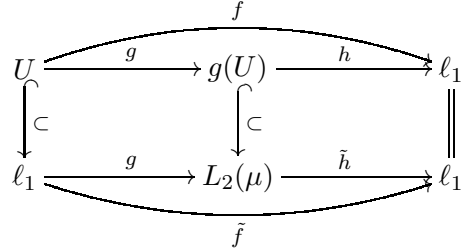
$$
\xymatrix @C=4.7pc {
U \ar@{^{(}->}[d]^{\subset} \ar[r]^g \ar@/^1.3pc/[rr]^f &g(U)\ar@{^{(}->}[d]^{\subset} \ar[r]^h &\ell_1\ar@{=}[d]\\
\ell_1 \ar[r]^g \ar@/_1.2pc/[rr]_{\tilde f} &L_2(\mu) \ar[r]^{\tilde h}  &\ell_1
}
$$

Now we define a map $h: g(U) \to \ell_1$ as $h(y) = f(g^{-1}(y))$.
The Lipschitz constant of $h$ is at most 
$\|f\|_{Lip} \|g^{-1}\|_{Lip} = O(\sqrt{\log k}\log\log k)$. We extend $h$ to a map $\tilde h: L_2(\mu) \to \ell_1$ such that 
$\|\tilde h\|_{Lip} \leq e_k(\ell_2,\ell_1) \|h\|_{Lip} = 
O(e_k(\ell_2,\ell_1) \sqrt{\log k}\log\log k).$
We finally define $\tilde f (x) = \tilde h(g(x))$.
For every $p\in U$, $\tilde f(p) =  \tilde h(g(p)) = h(g(p)) = f(p)$;
$\|\tilde f\|_{Lip} \leq \|\tilde h\|_{Lip} \|g\|_{Lip} = O(e_k(\ell_2,\ell_1) \sqrt{\log k}\log\log k)$.
This concludes the proof.
\end{proof}

\begin{corollary}
If Question~\ref{qst:ball} has a positive answer then there exist $\tilde O(\sqrt{\log k})$
cut sparsifiers. On the other hand, any lower bound on cut sparsifiers better
than $\tilde \Omega(\sqrt{\log k})$ would imply a negative answer to Question~\ref{qst:ball}.
\end{corollary}

\begin{remark} There are no pairs of Banach spaces $(X,Y)$ for which
$e_k(X,Y)$ is known to be greater than $\omega(\sqrt{\log k})$ (see 
e.g. \citeasnoun{LN}).
If indeed $e_k(X,Y)$ is always $O(\sqrt{\log k})$ then there
exist $O(\sqrt{\log k})$-quality metric sparsifiers.
\end{remark}

\section{Certificates for Quality of Sparsification} \label{sec:discussion}
In this section, we show that there exist
``combinatorial certificates'' for cut and metric sparsification that certify 
that $Q_k^{cut} \geq Q$ and $Q_k^{metric} \geq Q$.

\begin{definition}
A $(Q,k)$-certificate for cut sparsification is a tuple $(G, U, \mu_1, \mu_2)$ where $G = (V,\alpha)$
is a graph (with non-negative edge weights $\alpha$), $U\subset V$ is a subset of $k$ terminals, 
and $\mu_1$ and $\mu_2$ are distributions of cuts on $G$ such that for some (``scale'') $c > 0$
\begin{align*}
\Probb{S\sim\mu_1}{p\in S, q\notin S} &\leq c \Probb{S\sim \mu_2}{p\in S, q\notin S} &\forall p,q\in U,\\
\Exp_{S\sim\mu_1} \minext_{U \to V}(\delta_S,\alpha) &\geq c \cdot Q\cdot \Exp_{S\sim\mu_2} \minext_{U \to V}(\delta_S,\alpha) > 0,
\end{align*}
where $\minext_{U \to V}(\delta_S,\alpha)$ is the cost of the minimum cut in $G$ that separates $S$ and $U\setminus S$
(w.r.t. to edge weights $\alpha$).

Similarly, a $(Q,k)$-certificate for metric sparsification is a tuple $(G, U, \{d_i\}_{i=1}^{m_1})$
where $G = (V,\alpha)$ is a graph (with non-negative edge weights $\alpha$), $U\subset V$ is a subset of $k$ terminals, 
and $\{d_i\}_{i=1}^{m}$ is a family of metrics on $U$ such that
$$
\sum_{i=1}^m\minext_{U \to V}(d_i,\alpha) \geq Q \minext_{U \to V}\Bigl(\sum_{i=1}^m d_i,\alpha\Bigr)  > 0.
$$
\end{definition}

\begin{theorem}
If there exists a $(Q,k)$-certificate for cut or metric sparsification, then $Q_k^{cut} \geq Q$ or $Q_k^{metric} \geq Q$, respectively.
For every $k$, there exist $(Q_k^{cut},k)$-certificate for cut sparsification, and $(Q_k^{metric}-\varepsilon,k)$-certificate
for metric sparsification (for every $\varepsilon > 0$).
\end{theorem}
\begin{proof}
Let $(G, U, \mu_1, \mu_2)$ be a $(Q,k)$-certificate for cut sparsification. Let $(U,\beta)$ be a $Q_k^{cut}$-quality cut sparsifier for $G$.
Then 
\begin{align*}
\Exp_{S\sim\mu_1} & \minext_{U \to V}(\delta_S,\alpha) 
\leq \Exp_{S\sim\mu_1} \sum_{p\in S, q\in U\setminus S}\beta_{pq}
= \sum_{p,q\in U}\beta_{pq} \Probb{S\sim\mu_1}{p\in S, q\in U\setminus S} \\
&\leq 
c \sum_{p,q\in U}\beta_{pq} \Probb{S\sim\mu_2}{p\in S, q\in U\setminus S} =
\Exp_{S\sim\mu_2} c\sum_{p\in S, q\in U\setminus S}\beta_{pq} \leq c\cdot Q_k^{cut}\cdot 
\Exp_{S\sim\mu_2} \minext_{U \to V}(\delta_S,\alpha).
\end{align*}
Therefore, $Q_k^{cut} \geq Q$.

Now, let $(G, U, \{d_i\}_{i=1}^{m})$ be a $(Q,k)$-certificate for metric sparsification. Let 
$(U,\beta)$ be a $Q_k^{metric}$-quality metric sparsifier for $G$. Then
\begin{align*}
\sum_{i=1}^{m} \minext_{U \to V}(d_i,\alpha) &\leq 
\sum_{i=1}^{m}\sum_{p,q\in U}\beta_{pq} d_i(p,q)
= \sum_{p,q\in U}\beta_{pq}\sum_{i=1}^{m} d_i(p,q) \\
&\leq Q_k^{metric}\minext_{U \to V}\Bigl(\sum_{i=1}^md_i,\alpha\Bigr).
\end{align*}
Therefore, $Q_k^{metric} \geq Q$.

The existence of $(Q_k^{cut},k)$-certificates for cut sparsification, and $(Q_k^{metric}-\varepsilon,k)$-certificates for metric sparsification
follows immediately from the duality arguments in Theorems~\ref{thm:extendcut} and~\ref{thm:extendmetric}. We omit the details in
this version of the paper.
\end{proof}

\section*{Acknowledgements}
We are grateful to William Johnson and Gideon Schechtman for notifying us that a lower bound of $\Omega(\sqrt{\log k}/\log\log k)$ on $e_k(\ell_1,\ell_1)$ follows from their joint work with Figiel~\cite{FJS} and for giving us a permission to present the proof in this paper.

\pagebreak

\appendix
\section{Flow Sparsifiers are Metric Sparsifiers}
We have already established (in Lemma~\ref{lem:MetricIsFlow})
that every metric sparsifier is a flow sparsifier. We now prove 
that, in fact, every flow sparsifier is a metric sparsifier.
We shall use the same (standard) dual LP for the concurrent multi-commodity
flow as we used in the proof of Lemma~\ref{lem:MetricIsFlow}.
Denote the sum $\sum_{r} d_Y(s_r,t_r)\dem_k$ by 
$\gamma (d_Y)$. Then the definition of flow sparsifiers can be reformulated 
as follows: The graph $(Y,\beta)$ is a $Q$-quality flow sparsifier
for $(X,\alpha)$, if for every linear functional 
$\gamma:\calD_Y \to \bbR$ with nonnegative coefficients,
$$\min_{d_X\in\calD_X: \gamma({d_X}|_Y)\geq 1} \alpha(d_X) \leq
\min_{d_Y\in\calD_Y: \gamma(d_Y)\geq 1} \beta(d_Y)
\leq Q\times \min_{d_X\in\calD_X: \gamma({d_X}|_Y)\geq 1} \alpha(d_X).$$

\begin{lemma}\label{lem:FlowIsMetric}
Let $(X,\alpha)$ be a weighted graph and let $Y\subset X$ be a subset of vertices.
Suppose, that $(Y,\beta)$ is a $Q$-quality flow sparsifier, then 
$(Y,\beta)$ is also a $Q$-quality metric sparsifier.
\end{lemma}
\begin{proof}
We need to verify that for every $d_Y\in \calD_Y$,
$$
\minext_{Y \to X}(d_Y,\alpha) \leq \beta (d_Y) \leq 
Q\times \minext_{Y \to X}(d_Y,\alpha).$$

Verify the first inequality. Suppose that it does not hold for some $d^*_Y\in \calD_Y$.
Let $\widetilde{\calD}_Y = \{d_Y\in \calD_Y: \minext_{Y \to X}(d_Y,\alpha) \leq \beta (d^*_Y)\}.$
The set $\widetilde{\calD}_Y$ is closed (and compact, if the graph is connected) and convex (because $\minext$ is a convex function of the first variable). Since $\minext_{Y \to X}(d^*_Y,\alpha) > \beta (d^*_Y)$, $d^*_Y\notin \widetilde{\calD}_Y$.
Hence, there exists a linear functional $\gamma$ separating $d^*_Y$ 
from $\widetilde{\calD}_Y$. That is, $\gamma (d^*_Y) \geq 1$, but for every 
$d_Y\in \widetilde{\calD}_Y$, $\gamma (d_Y) < 1$. We show in Lemma~\ref{lem:poscoef}, that there exists such 
$\gamma$ with nonnegative coefficients. Then, by the definition of the flow sparsifier,
$$\min_{d_X\in\calD_X: \gamma({d_X}|_Y)\geq 1} \alpha(d_X) \leq
\min_{d_Y\in\calD_Y: \gamma(d_Y)\geq 1} \beta(d_Y).$$
But, the left hand side
$$\min_{d_X\in\calD_X: \gamma({d_X}|_Y)\geq 1} \alpha(d_X) = 
\min_{d_Y\in\calD_Y:\gamma(d_Y)\geq 1} \minext_{Y\to X}(d_Y,\alpha) \geq 
\min_{d_Y\notin\widetilde{\calD}_Y} \minext_{Y\to X}(d_Y,\alpha) > \beta(d^*_Y);$$
and the right hand side is at most $\beta(d^*_Y)$, since $\gamma(d^*_Y)\geq 1$. We get 
a contradiction.

Verify the second inequality. Let $\gamma (d_Y) = \beta (d_Y)/\beta (d^*_Y)$. 
By the definition of the flow sparsifier,
$$\min_{d_Y\in\calD_Y: \gamma(d_Y)\geq 1} \beta(d_Y)
\leq Q\times \min_{d_X\in\calD_X: \gamma({d_X}|_Y)\geq 1} \alpha(d_X).$$
The left hand side is at least $\beta (d^*_Y)$ (by the definition of $\gamma$).
Thus, for every $d_X\in\calD_X$ satisfying $\gamma({d_X}|_Y)\geq 1$, and particularly, for $d_X$ equal 
to the minimum extension of $d_Y$,
$Q\times \alpha(d_X) \geq \beta (d^*_Y)$.
\end{proof}

\begin{lemma}[Minimum extension is monotone]\label{lem:mono}
Let $X$ be an arbitrary set, $Y\subset X$, and $\alpha_{ij}$ be a nonnegative set of weights
on pairs $(i,j)\in X\times X$. Suppose that a metric $d^{*}_Y\in\calD_Y$ dominates 
metric $d^{**}_Y\in\calD_Y$ i.e., $d^{*}_Y(p,q)\geq d^{**}_Y(p,q)$ for every $p,q\in Y$.
Then,
$$\minext_{Y\to X}(d^{*}_Y, \alpha) \geq \minext_{Y\to X}(d^{**}_Y, \alpha).$$
\end{lemma}
\begin{proof}[Proof sketch]
Let $d^{*}_X$ be the minimum extension of $d^{*}_Y$. Consider the
distance function 
$$d^{**}_X (i,j) = 
\begin{cases}
d_Y^{**}(i,j),&\text{if } i,j \in Y;\\
d_X^{*}(i,j),&\text{otherwise}.
\end{cases}
$$
The function $d^{**}_X (i,j)$ does not necessarily satisfy the 
triangle inequalities. However, the shortest path metric $d^{s}_X$ induced by $d^{**}_X$ does satisfy the triangle inequalities, and 
is an extension of $d^{**}_Y$. Since, $d^{*}_X (i,j) \geq d^{**}_X (i,j) \geq d^{s}_X (i,j)$
for every $i,j\in X$, 
$$\minext_{Y\to X}(d^{*}_Y, \alpha) = \alpha (d^*_X) \geq \alpha (d^{s}_X)
\geq \minext_{Y\to X}(d^{**}_Y, \alpha).$$
\end{proof}

\begin{lemma}\label{lem:poscoef}
Let $\widetilde{\calD}_Y = \{d_Y\in \calD_Y: \minext_{Y\to X}(d_Y,\alpha) \leq 1 \}$, and 
$d^*_Y\in \calD_{Y}\setminus \widetilde{\calD}_Y$. Then, there exists a linear
functional 
$$\gamma (d_Y) = \sum_{p,q\in \calD_Y} \gamma_{pq}  d_Y(p,q),$$
with nonnegative coefficients $\gamma_{pq}$ separating $d^*_Y$ from 
$\widetilde{\calD}_Y$, i.e., $\gamma(d^*_Y)\geq 1$, but for every 
$d_Y\in \widetilde{\calD}_Y$, $\gamma(d_Y) < 1$.
\end{lemma}
\begin{proof}
Let $\Gamma$ be the set of linear functionals $\gamma$ with nonnegative coefficients
such that $\gamma(d^*_Y)\geq 1$. This set is convex. We need to show that there 
exists $\gamma\in \Gamma$ such that $\gamma(d_Y) < 1$ for every 
$d_Y \in \widetilde{\calD}_Y$. By the \citeasnoun{Neumann} minimax theorem,
it suffices to show that for every $d_Y^{**}\in \widetilde{\calD}_Y$, there exists
a linear functional $\gamma\in \Gamma$ such that $\gamma (d_Y^{**}) < 1$. By Lemma~\ref{lem:mono}, since
$$\minext_{Y\to X} (d_Y^{**},\alpha) < 1 \leq \minext_{Y\to X} (d^{*}_Y,\alpha),$$
there exist $p,q\in Y$, such that $d_Y^{**}(p,q) < d_Y^{*}(p,q)$. 
The desired linear functional is $\gamma (d_Y) = d_Y(p,q)/d^{*}_Y(p,q)$.
\end{proof}
\section{Compactness Theorem for Lipschitz Extendability Constants}
In this section, we prove a compactness theorem for Lipschitz
extendability constants.
\begin{theorem}
Let $X$ be an arbitrary metric space and $V$ be a finite dimensional normed space.
Assume that for some $K$ and every $Z \subset \tilde Z \subset V$
with $|Z| = k$, $|\tilde Z| < \infty$, every map $f:Z\to V$ can be extended 
to a map $\tilde f:\tilde Z \to V$
so that $\|\tilde f\|_{Lip} \leq K \|f\|_{Lip}$. Then $e_k(X,V) \leq K$.
\label{thm:compactness}
\end{theorem}
\begin{proof}
Fix a set $Z$ and a map $f:Z\to V$. Without loss of generality we may assume that
$\|f\|_{Lip} = 1$. We shall construct a $K$-Lipschitz extension $\hat f:X\to V$ of
$f$.

Choose an arbitrary $z_0\in Z$. 
Consider the following topological 
space of maps from $X$ to $V$:
$${\cal F} = \{h:X\to V: \forall x\in X\, \|h(x) - f(z_0)\|_V \leq K d(z_0, x)\} 
\cong \prod_{x\in X} B_V(f(z_0), K d(z_0, x)),
$$
equipped with the product topology (the topology of pointwise convergence); 
i.e., a sequence of functions $f_i$ converges to $f$ if for every $x\in X$, $f_i(x) \to f(x)$.
Note that every ball $B_V(f(z_0), K d(z_0, x))$ is a compact set. By Tychonoff's theorem
the product of compact sets is a compact set. Therefore, $\cal F$ is also a compact set.

Let $M$ be the set of maps in $\cal F$ that extend $f$: 
$M = \{h\in {\cal F}: h(z) = f(z)\text{ for all } z\in Z\}$. Let $C_{x,y}$ (for $x,y \in X$)
be the set of 
functions in $\cal F$ that increase the distance between points $x$ and $y$
by at most a factor of $K$: $C_{x,y} = \{h\in {\cal F}: \|h(x) - h(y)\|_V 
\leq K d(x,y)\}$. Note that all sets $M$ and $C_{x,y}$ are closed.
We prove that every finite family of sets $C_{x,y}$ has a non-empty
intersection with $M$. Consider a finite family of sets:
$C_{x_1,y_1}, \dots, C_{x_n,y_n}$. Let $\tilde Z = Z\cup \bigcup_{i=1}^n\{x_i, y_i\}$.
By the condition of the theorem there exists a $K$-Lipschitz map 
$\tilde f: \tilde Z \to V$ extending $f$. Then 
$\tilde f\in \bigcap_{i=1}^n C_{x_i,y_i} \cap M$. Therefore,
$\bigcap_{i=1}^n C_{x_i,y_i} \cap M \neq \varnothing$.

Since every finite family of closed sets in $\{M, C_{x,y}\}$
has a non-empty intersection and $\cal F$ is compact,
all sets $M$ and $C_{x,y}$ have a non-empty intersection.
Let $\hat f \in M \cap \bigcap_{x,y\in X} C_{x,y}$. Since
$\hat f\in M$, $\hat f$ is an extension of $f$. Since
$\hat f \in C_{x,y}$ for every $x,y\in X$, the map
$\hat f$ is $K$-Lipschitz.
\end{proof}

\section{Lipschitz Extendability and Projection Constants}
\label{sec:JL-overview}
In Section~\ref{subsec:proj-const}, we use the following theorem of~\citeasnoun{JL}.
In their paper, however, this theorem is stated in a slightly different form. 
We sketch here the original proof of Johnson and Lindenstrauss for completeness. 

\begin{theorem}[\citeasnoun{JL}, Theorem 3]\label{thm:JL}
Let $V$ be a Banach space, $L \subset V$ be a $d$-dimensional 
subspace of $V$, and $U$ be a finite dimensional normed space. 
Then every linear operator $T:L \to U$, with  $\|T\|\|T^{-1}\| = O(d)$, 
can be extended to a linear operator $\tilde T: V \to U$
so that $\|\tilde T\| = O(e_k(V, U)) \|T\|$, where $k$ is such that 
$d \leq {c\log k}/{\log\log k}$ (where $c$ is an absolute constant).

In particular, for $U=L$, the identity operator $I_L$ on $L$ can be extended to a projection
$P: V\to L$ with $\|P\| \leq O(e_k(V, L))$. Therefore, $\lambda(L, V) = O(e_k(V, L))$.
\end{theorem}

First, we address a simple case when $e_k(V, U) \geq \sqrt{d}$.
By the Kadec--Snobar theorem there exists a projection $P_L$ from
$V$ to $L$ with $\|P_L\| \leq \sqrt{d}$. Therefore, $TP_L$ is an extension of $T$
with the norm bounded by $\sqrt{d}\|T\|$ and we are done. So we assume below that 
$e_k(V, U) \leq \sqrt{d}$
  
We construct the extension $\tilde T$ in several steps. Denote $\alpha = \|T\|\|T^{-1}\|$.
First, we choose an $\varepsilon$-net $A$ of size at most $k-1$ on the unit
sphere $S(L) = \{v\in L:\|v\|_{V} = 1\}$
for $\varepsilon \sim 1/(\alpha\log^2 k)$ (to be specified later).
\begin{lemma}[\citeasnoun{JL}, Lemma~3]
If $L$ is a $d$-dimensional normed space and $\varepsilon > 0$ then
$S(L)$ admits an $\varepsilon$-net of cardinality at most $(1 + 4/\varepsilon)^d$.
\end{lemma}
Let $T_1$ be the restriction of $T$ to $A\cup \{0\}$.
Let $S(V) = \{v\in V:\|v\|_{V} = 1\}$.
By the definition of the Lipschitz extendability constant $e_k(V, U)$,
there exists an extension $T_2:S(V) \to U$ of $T_1$ with $\|T_2\|_{Lip}  \leq e_k(V, U)  \|T_1\|_{Lip} 
\leq e_k(V, U) \|T\|$.
Now we consider the \textit{positively homogeneous extension} 
$T_3:V\to U$ of $T_2$ defined as 
$$T_3(v) = \|v\|_V T_2\left(\frac{v}{\|v\|_V}\right).$$
The following lemma gives a bound on the norm of $T_3$.
\begin{lemma}[\citeasnoun{JL}, Lemma~2]
Suppose that $V$ and $U$ are normed spaces, and 
$f: S(V) \cup \{0\} \to U$ is a Lipschitz map with $f(0) = 0$. 
Then the positively homogeneous extension
$\tilde f$ of $f$ is Lipschitz and 
$$\|\tilde f\|_{Lip} \leq 2 \|f\|_{Lip} + \sup_{v\in S(V)} \|f(v)\|_{U}.$$
\end{lemma}
Since $T_2(0) = 0$ and $\|T_2 \|_{Lip} \leq e_k(V,U)\|T\|$, 
$\sup_{v\in S(V)} \|T_2v\|_{V} \leq \|T_2\|_{Lip}\leq e_k(V,U)\|T\|$.
Therefore, $\|T_3\|_{Lip} \leq 3e_k(V,U)\|T\|$.
Now we prove that there exists a Lipschitz map $T_4:V \to U$,
whose restriction to $L$ is very close to $T$.
We apply the following lemma to $F = T_3$ and obtain a map $T_4= \tilde F:V\to U$.
\begin{lemma}[\citeasnoun{JL}, Lemma 5]
Suppose $L\subset V$ and $U$ are Banach spaces
with $\dim L = d < \infty$, $F:V \to U$
is Lipschitz with $F$ positively homogeneous (i.e.
$F(\lambda v) = \lambda F(v)$ for $\lambda > 0$,
$v\in V$) and $T: L \to V$ is linear. Then there is
a positively homogeneous map
$\tilde F:V \to U$ which satisfies
\begin{itemize}
\item $\|\tilde F|_{L} - T\|_{Lip} \leq (8d + 2) \sup_{v\in S(L)} \|F(v) - T(v)\|_{V}$,
\item $\|\tilde F\|_{Lip} \leq 4 \|F\|_{Lip}.$
\end{itemize}
\end{lemma}

\noindent Note that for every $u\in S(L)$ there exists $v\in A$ with $\|u-v\|_{V}\leq \varepsilon$.
Therefore,
\begin{align*}
\|T_3u - Tu\|_{V} &\leq \|T_3u - T_3v\|_{V} + \|T_3v - Tv\|_{V} + \|Tv-Tu\|_{V}\\
&\leq \|T_3\|_{Lip} \cdot \varepsilon + 0 + \|T\| \varepsilon \leq (3e_k(V, U) + 1)\|T\| \varepsilon.
\end{align*}
Hence, 
$$\|T_4|_{L} - T\|_{Lip} \leq (8d + 2)(3e_k(V, U) + 1)\|T\| \varepsilon
\leq 40 d e_k(V, U) \|T\|  \varepsilon,$$
and $\|T_4\|_{Lip} \leq 12 e_k(V, U)\|T\|$.
Finally, we approximate $T_4$ with a linear bounded map 
$T_5: V \to U$, whose restriction to $L$ is very close to $T$.
\begin{lemma}[\citeasnoun{JL}, Proposition 1]
Suppose $L \subset V$ and $U$ are Banach spaces, $U$ is a reflexive 
space, $f:V\to L$ is Lipschitz, and 
$T:L\to U$ is bounded, linear. Then there is a linear operator $F:V\to U$ that satisfies 
$\|F\| \leq \|f\|_{Lip}$ and 
$\|F|_{L} - T\|_{L\to U} \leq \|f_{L} - U\|_{Lip}$.
\end{lemma}
Since the space $U$ is finite dimensional, it is reflexive. We apply the lemma
to $f=T_4$ and obtain a linear operator $T_5:V\to U$ such that
$\|T_5\| \leq  12 e_k(V, U)\|T\|$ and
$$\|T_5|_{L} - T\|_{L\to U} \leq 40 d e_k(V, U) \|T\|\varepsilon.$$
Let $P:U\to T(L)$ be a projection of $U$ on $T(L)$ with $\|P\| \leq \sqrt{d}$ 
(such projection exists by the Kadec--Snobar theorem).
Consider a linear operator $\phi=T_5 T^{-1}P + (I_U - P)$ from $U$ to $U$.
Note that for every $u\in U$,
\begin{align*}
\|\phi u - u\|_U &= \|T_5T^{-1} Pu - Pu\|_U = \|T_5T^{-1} Pu - T T^{-1} Pu\|_U \leq  40 d e_k(V, U) \|T\|\varepsilon \cdot \|T^{-1} Pu \|_U\\
&\leq 40 d e_k(V, U) \|T\|\varepsilon \cdot \sqrt{d} \|T^{-1}\| \|u\|_U \leq  40 \alpha d^{2} \varepsilon \|u\|_U
\end{align*}
(we used that $e_k(V, U)\leq \sqrt d$ and $\|P\|\leq \sqrt{d}$).
We choose $\varepsilon \sim 1/(\alpha \log^{2} k)$ so that
$40 \alpha d^{2}\varepsilon <1/2$. Then $\|\phi- I_{U}\| \leq 1/2$. Thus
$\phi$ is invertible:
$$\phi^{-1} = (I_{U} - (I_{U} - \phi))^{-1} = \sum_{i=0}^\infty (I_{U} - \phi)^k,$$
and
$$\|\phi^{-1}\| \leq \sum_{i=0}^\infty \|I_{U} - \phi\|^k \leq 2.$$
Finally, we let $\tilde T = \phi^{-1} T_5$.
Note that for every $u\in L$, $\phi Tu = T_5u=\phi \tilde T u$, thus 
$\tilde T$ is an extension of $T$. The norm of $\tilde T$ is bounded by
$\|\phi\| \|T_5\| \leq 24 e_k(V, U)\|T\|$.

\section{Improved Lower Bound on $e_k(\ell_1, \ell_1)$} \label{sec:improved-l1}
After a preliminary version of our paper appeared as a preprint, Johnson and Schechtman notified us that
our lower bound of $\Omega(\sqrt[4]{\log k/\log\log k})$ on $e_k(\ell_1, \ell_1)$ can be 
improved to $\Omega(\sqrt{\log k}/\log\log k)$. This result follows from the paper of \citeasnoun{FJS} that studies factorization of operators to $L_1$ through $L_1$. With the permission of Johnson and Schechtman, we present this result below.

Before we proceed with the proof, we state the result of \citeasnoun{FJS}.

\begin{theorem}[Corollary 1.5, \citeasnoun{FJS}]\label{lem:FJS}
Let $X$ be a $d$-dimensional subspace of $L_1(R, \mu)$ (a set of real valued functions on $R$ with the $\|\cdot\|_1$ norm). 
Suppose that for every $f\in X$ and every $2\leq r < \infty$, $\|f\|_r \leq C \sqrt{r} \|f\|_1$ (where $C$ is some constant not depending on $f$ and $r$). Let $w:X \to \ell_1^m$ and $u:\ell_1^m \to L_1(R, \mu)$ be linear operators such that
$uw = I_{X}$ is the identity operator on $X$. Then 
$$\rank u \geq 2^{\Delta d} \text{ where\ } \Delta= \frac{1}{(16 C d_{BM}(X, \ell_2^d) \|w\| \|u\|)^2}.$$
\end{theorem}

\begin{corollary} $e_k(\ell_1,\ell_1) = \Omega\left(\sqrt{\log k}/\log\log k\right)$. 
\end{corollary}
\begin{proof}
Denote $d = c \log k/\log\log k$, where $c$ is the constant from Theorem~\ref{thm:JL}.
Consider $U= \ell_1^{2d}$. By Kashin's theorem~\cite{Kashin}, there exists an ``almost Euclidean'' $d$-dimensional subspace
$X'$ in $U$, that is, a subspace $X'$ such that 
$$c_1 \|x\|_1 \leq \sqrt{d}\, \|x\|_2 \leq c_2 \|x\|_1$$ for every $x\in X'$ 
(and some positive absolute constants $c_1$ and $c_2$). Let $R = \{\pm1\}^{2d} \subset U$ be a $2d$-dimensional hypercube,
$\mu$ be the uniform probabilistic measure on $R$ and $V = L_1(R,\mu)$. We consider a natural embedding $u'$ of $X'$ into $V$:
each vector $x\in X'$ is mapped to a function $u'(x)\in V$ defined by $u'(x): y \mapsto \langle x, y\rangle$. 
Recall that by the Khintchine inequality,
$$A_p \|x\|_2 \leq \|u'(x)\|_p \equiv (\EE{y\in R}{|\langle x, y\rangle|^p})^{1/p} \leq B_p \|x\|_2,$$
where $A_p$ and $B_p$ are some positive constants. In particular, \citeasnoun{Haagerup} proved that the inequality holds for $p=1$, with
$A_1 = \sqrt{1/2}$ and $B_1=\sqrt{2/\pi}$, and, for $p\geq 2$, with $A_p = 1$ and 
$$B_p = 2^{1/2 -1/p} \left(\left.\Gamma\left(\frac{p+1}{2}\right)\right/\Gamma\left(\frac32\right)\right)^{1/p} = (1+ o(1)) \sqrt{\frac{p}{e}} $$
(the $o(1)$ term tends to $0$ as $p$ tends to infinity).
Let $X\subset L(R,\mu)$ be the image of $X'$ under $u'$. Observe that $u'$ is a $(2c_2/c_1)$-isomorphism between $(X',\|\cdot\|_1)$ and $(V,\|\cdot\|_1)$.
Indeed,
\begin{align*}
\|u'(x)\|_1 &\leq B_1 \|x\|_2 \leq \sqrt{2} \cdot c_2\|x\|_1 /\sqrt{\pi d}, \\
\|u'(x)\|_1 &\geq A_1 \|x\|_2 \geq  c_1  \|x\|_1/ \sqrt{2d}.
\end{align*}
Denote $w=(u')^{-1}$. Then $\|u'\| \|w\| \leq 2 c_2/(\sqrt{\pi} c_1) < 2c_2/c_1$.

By Theorem~\ref{thm:JL}, there exists a linear extension $u:U\to V$ of $u'$ to $U$ with $\|u\| = O(e_k(\ell_1,\ell_1)) \|u'\|$.
We are going to apply Lemma~\ref{lem:FJS} to maps $u$ and $w$ and get a lower bound on $\|u\|$ and, consequently,
on $e_k(\ell_1,\ell_1)$. To do so, we verify that for every $f\in X$,
$\|f\|_r = O(\sqrt{r})$. Indeed, if $f = u'(x)$, we have
$$\|f\|_r \leq B_r \|x\|_2 \leq B_r \|f\|_1 /A_1 = \sqrt{\frac{2}{e}}\cdot \sqrt{r} \cdot \|f\|_1 (1+o(1)).$$
%
Note that $\rank u \leq \dim U = 2d$ and $d_{BM}(X, \ell_2^d) \leq B_1/A_1 = 2/\sqrt{\pi}$.
By Lemma~\ref{lem:FJS}, we have
$$\Delta \equiv \frac{1}{(16 C d_{BM}(X, \ell_2^d) \|w\| \|u\|)^2} \leq \frac{\log_2 {2d}}{d}.$$
Therefore,
$$\|u\| \geq \Omega\left(\sqrt{\frac{d}{\log d}} \right) \frac{1}{\|w\|} = \Omega\left(\sqrt{\frac{d}{\log d}}\right) \|u'\|.$$
We conclude that 
$$e_k(\ell_1,\ell_1) \geq \Omega\left(\|u\|/ \|u'\|\right) \geq \Omega\left(\sqrt{\frac{d}{\log d}}\right)
 = \Omega\left(\frac{\sqrt{\log k}}{\log\log k}\right).$$
\end{proof}
\end{document}